\documentclass{article}
\usepackage{amsfonts}
\usepackage{amsmath}
\usepackage{fullpage}
\usepackage{graphicx}
\usepackage{subfigure}
\usepackage{color,hyperref,cite}
\usepackage{float}
\usepackage{algorithm}
\usepackage[noend]{algorithmic}
\usepackage[capitalise]{cleveref}
\usepackage{tabularx}
\usepackage{xspace}
\usepackage{tikz}

\newtheorem{invariant}{Invariant}
\newtheorem{col-rule}{Coloring Rule}
\newtheorem{theorem}{Theorem}

\newtheorem{corollary}{Corollary}

\newtheorem{definition}{Definition}
\newtheorem{observation}{Observation}

\newtheorem{lemma}{Lemma}

\newenvironment{proof}[1][Proof]{\noindent\textbf{#1.} }{\ \rule{0.5em}{0.5em}}

\newcommand{\problemdefshort}[3]{
    \begin{center}
    \fbox{ 
    \begin{minipage}{0.95\textwidth}
      \noindent
      \normalsize\textsc{#1}
      
      \vspace{1pt}
      \setlength{\tabcolsep}{3pt}
      \renewcommand{\arraystretch}{1.0}
      \begin{tabularx}{\textwidth}{@{}lX@{}}
	\normalsize\textbf{Input:} 	& \normalsize#2 \\
	\normalsize\textbf{Task:} 	& \normalsize#3 
      \end{tabularx}
    \end{minipage}
    }
    \end{center}
}

\newcommand{\dist}{\textup{dist}}

\begin{document}

\title{Exact and Approximate Algorithms \\for Computing a Second Hamiltonian Cycle\thanks{Supported by 
the NeST initiative of the EEE/CS School of the University of Liverpool and by the EPSRC grants EP/P020372/1 and EP/P02002X/1.}}
\author{Argyrios Deligkas\thanks{Department of Computer Science, Royal Holloway University of London, UK. 
Email: \texttt{argyrios.deligkas@rhul.ac.uk}} \and 
George B. Mertzios\thanks{Department of Computer Science, Durham University, UK.
Email: \texttt{george.mertzios@durham.ac.uk}} \and 
Paul G. Spirakis\thanks{Department of Computer Science, University of Liverpool, UK, and 
Computer Engineering \& Informatics Department, University of Patras, Greece. 
Email: \texttt{p.spirakis@liverpool.ac.uk}} \and 
Viktor Zamaraev\thanks{Department of Computer Science, University of Liverpool, Liverpool, UK. 
Email: \texttt{viktor.zamaraev@liverpool.ac.uk}}}

\date{\vspace{-0.5cm}}
\maketitle

\begin{abstract}
In this paper we consider the following total functional problem: 
Given a cubic Hamiltonian graph~$G$ and a Hamiltonian cycle $C_0$ of $G$, 
how can we compute a second Hamiltonian cycle $C_1 \neq C_0$ of~$G$? 
Cedric Smith and William Tutte proved in 1946, using a non-constructive parity argument, that such a second Hamiltonian cycle always exists. 
Our main result is a deterministic algorithm which computes the second Hamiltonian cycle in 
$O(n\cdot 2^{0.299862744n}) = O(1.23103^n)$ time and in linear space, 
thus improving the state of the art running time of $O^{*}(2^{0.3 n}) = O(1.2312^n)$ 
for solving this problem (among deterministic algorithms running in polynomial space). 
Whenever the input graph $G$ does not contain any induced cycle~$C_{6}$ on 6~vertices, the running time becomes $O(n\cdot 2^{0.2971925n}) = O(1.22876^n)$. 
Our algorithm is based on a fundamental structural property of Thomason's lollipop algorithm, which we prove here for the first time. 
In the direction of approximating the length of a second cycle in a (not necessarily cubic) Hamiltonian graph $G$ with a given Hamiltonian cycle $C_0$ 
(where we may not have guarantees on the existence of a second Hamiltonian cycle), 
we provide a linear-time algorithm computing a second cycle with length at least $n - 4\alpha (\sqrt{n}+2\alpha)+8$, 
where $\alpha = \frac{\Delta-2}{\delta-2}$ and $\delta,\Delta$ are the minimum and the maximum degree of the graph, respectively. 
This approximation result also improves the state of the art.
\end{abstract}

\section{Introduction}
\label{intro-sec}

Graph Hamiltonicity problems are among the most fundamental problems in theoretical computer science. 
Problems related to Hamiltonian paths and Hamiltonian cycles have attracted a tremendous amount 
of work over the years, see for example the recent survey of Gould~\cite{G14} and the references therein. 
Deciding whether a given graph has a Hamiltonian cycle, i.e.~a cycle that contains each vertex once, 
was among Karp's 21 NP-hard problems~\cite{K72}. On the other hand, there are several exponential-time algorithms 
for computing a Hamiltonian cycle or a solution to the Traveling Salesman Problem (TSP), which is a direct generalization of the Hamiltonian cycle problem. 
The first algorithms for the problem were based on dynamic programming and required $O(n^2 2^n)$ time~\cite{Bel62, HK62}. 
One of the next major improvements came decades later by Eppstein~\cite{Eppstein07} who showed that 
a Hamiltonian cycle in a graph of degree at most three with $n$ vertices 
can be computed in $O(2^{\frac{n}{3}}) \approx 1.26^n$ time and linear space; 
at the same time the algorithm can also compute an optimum solution for TSP on such graphs. 
The algorithm of Eppstein works by forcing specific edges of the graph which must be part of any generated cycle; 
a variation of this algorithm can also enumerate all Hamiltonian cycles in a graph of degree at most three in $O(2^{\frac{3n}{8}})$ time~\cite{Eppstein07}. 
After that, there has been a series of improvements on the running time for TSP and the Hamiltonian cycle problem in degree-three graphs. 
In this direction there are two different lines of research, one for algorithms using polynomial space and one for algorithms using exponential space. 
With respect to algorithms using polynomial space, the most recent results are an $O(1.2553^n)$-time algorithm by Li{\'s}kiewicz and Schuster~\cite{LS14} 
and an $O^{*}(2^{0.3 n}) = O(1.2312^n)$-time algorithm by Xiao and Nagamochi~\cite{XiaoNagamochi16}, where $O^{*}(\cdot)$ suppresses polynomial factors. 
For bounded-degree graphs, it is known by Bj\"orklund et al.~\cite{BjorklundHKK12} that TSP can be solved in $O^{*}((2-\varepsilon)^n)$ time, 
where $\varepsilon>0$ only depends on the maximum degree of the input graph. 
Furthermore, for general graphs there exists a Monte Carlo algorithm for computing a Hamiltonian cycle 
with running time $O^{*}(1.657^n)$, given by Bj\"orklund~\cite{B14}. 
By allowing exponential space, the running time for solving TSP on degree-three graphs can be improved further to $O^{*}(1.2186^n)$~\cite{BHKN}, 
while a Hamiltonian cycle can also be detected in $O^{*}(1.1583^n)$ time using a Monte Carlo algorithm~\cite{CyganKN13}. 
In our paper we focus on algorithms running in polynomial space.

On the other hand, using a non-constructive parity argument, Cedric Smith and William Tutte~\cite{Tu46} proved in 1946 that, 
for any fixed edge in a cubic (i.e.~3-regular) graph $G$, there exists an even (potentially zero) number of Hamiltonian cycles through this edge. 
Thus, the existence of a first Hamiltonian cycle guarantees the existence of a \emph{second} one too, 
and this allows us to define the following total functional problem~\cite{pap94}.

\problemdefshort{Smith}
{A cubic Hamiltonian graph $G$ and a Hamiltonian cycle $C_0$ of $G$.}
{Compute a second Hamiltonian cycle $C_1 \neq C_0$ of $G$.}

It is easy to see that any algorithm $\mathcal{A}$ for the Hamiltonian cycle (decision) problem on graphs with maximum degree three 
can be trivially adapted to solve \textsc{Smith} as follows: 
for every edge $e$ of the initial Hamiltonian cycle $C_0$, run $\mathcal{A}$ on $G\setminus e$, 
i.e.~on the graph obtained by removing $e$ from $G$. Then, as a second Hamiltonian cycle $C_1 \neq C_0$ always exists, 
at least one of these $n$ calls of $\mathcal{A}$ will return such a cycle $C_1$. That is, \textsc{Smith} can be solved in 
$n\cdot T(\mathcal{A})$ time, where $T(\mathcal{A})$ is the worst-case running time of $\mathcal{A}$ on input graphs with $n$ vertices. 
Similarly, any algorithm $\mathcal{A}'$ which computes the parity of the number of Hamiltonian cycles in a given graph can be also 
used as a subroutine to solve \textsc{Smith}. Such an algorithm $\mathcal{A}'$, 
which runs in time $O(1.619^n)$ and in polynomial space, was given by Bj\"orklund and Husfeldt~\cite{BjorklundH13} for directed graphs, 
but the result carries over to undirected graphs as well.

Thomason~\cite{T78} was the first one who provided an algorithm, known as the \emph{lollipop algorithm}, for \textsc{Smith}. 
This algorithm starts from the given Hamiltonian cycle $C_0$ of $G$ and creates a sequence of distinct Hamiltonian paths 
where the last of these Hamiltonian paths trivially augments to a different Hamiltonian cycle of $G$. 
This algorithm was actually used by Papadimitriou to place \textsc{Smith} within the complexity class PPA~\cite{pap94}. 
Although Thomason's lollipop algorithm is well-known for decades, the
internal structure of the algorithm's execution on cubic Hamiltonian graphs
remains so far mostly unclear and not well understood.  
In an attempt to construct worst-case instances for the lollipop algorithm, Cameron proved in~2001~\cite{Cameron01} 
that on a specific family of cubic graphs (which is a variation of the
family introduced by Krawczyk~\cite{Krawczyk99}) the lollipop algorithm runs
in time at least $2^{cn}$, for some constant $c$. 
Thus, the state of the art running time (using polynomial space) for computing a second Hamiltonian cycle in \textsc{Smith} 
is to use the best known algorithm for the Hamiltonian cycle problem in cubic graphs 
which runs in $O^{*}(2^{0.3 n})$~\cite{XiaoNagamochi16}. 
However, a tantalizing longstanding question is whether the knowledge of the first Hamiltonian cycle $C_0$ 
\emph{strictly helps} to reduce the running time for computing a second Hamiltonian cycle $C_1$.
In this paper we provide evidence for the \emph{affirmative} answer to this question.

A relaxation of \textsc{Smith} is, given a Hamiltonian cycle $C_0$, to efficiently compute a second cycle (different than $C_0$) 
that is large enough. 
This relaxed problem becomes more meaningful for graphs with degrees larger than three, as it is well known that 
uniquely Hamiltonian graphs (i.e.~graphs with a unique Hamiltonian cycle) exist, even when all vertices have 
degree three except two vertices which have degree four~\cite{fleischner94,entringer-swart80}. 
For cubic Hamiltonian graphs, Bazgan, Santha, and Tuza~\cite{bazgan1999approximation} showed that 
the knowledge of the first Hamiltonian cycle $C_0$ algorithmically strictly helps to approximate the length of a second cycle. 
In fact, if $C_0$ is not given along with the input, 
there is no polynomial-time constant-factor approximation algorithm for finding a long cycle 
in cubic graphs, unless P=NP. In contrast, if $C_0$ is given, then for every $\varepsilon>0$ a cycle $C' \neq C_0$ 
of length at least $(1-\varepsilon)n$ can be found in $2^{O(1 / \varepsilon^2)}\cdot n$ time, 
i.e.~there is a linear-time PTAS for approximating the second Hamiltonian cycle~\cite{bazgan1999approximation}. 
The main ingredient in the proof of the latter result is an $O(n^{\frac{3}{2}} \log n)$-time algorithm which, given $G$ and~$C_0$, 
computes a cycle $C' \neq C_0$ of length at least $n-4\sqrt{n}$~\cite{bazgan1999approximation}.
In wide contrast to cubic graphs, for graphs of minimum degree at least three, only existential proofs are known for a second large cycle. 
In particular, Gir{\~a}o, Kittipassorn, and Narayanan recently proved with a \emph{non-constructive} argument 
that any $n$-vertex Hamiltonian graph with minimum degree at least $3$ 
contains another cycle of length at least $n - o(n)$~\cite{girao2019long}.

\medskip
\noindent\textbf{Our contribution.} In this paper we do the first attempt to understand 
the internal structure of the lollipop algorithm of Thomason~\cite{T78}. 
Our main result in this direction embarks from the following trivial observation, which is not specific to Thomason's algorithm or to cubic graphs.

\begin{observation}
\label{2-factor-symm-diff-obs}Let $G$ be a cubic Hamiltonian graph and let $C_{0},C_{1}$
be any two different Hamiltonian cycles of $G$. 
Then the symmetric difference $C_{0}\ \Delta \ C_{1}$ of the edges of the 
two cycles is a $2$-factor, i.e.~a collection of cycles in $G$.
\end{observation}

Although Observation~\ref{2-factor-symm-diff-obs} determines that the
symmetric difference of any two Hamiltonian cycles $C_{0}$ and $C_{1}$ is a
collection of cycles in $G$, it does not rule out the possibility that $%
C_{0}\ \Delta \ C_{1}$ contains more than one cycle. 
Our first technical contribution is that, for any given Hamiltonian cycle $C_{0}$, 
there exists at least one other Hamiltonian cycle $C_{1}$ such that 
$C_{0}\ \Delta \ C_{1}$ is \emph{connected}, i.e.~it contains exactly one cycle. 
More specifically, we prove that this holds for the particular Hamiltonian cycle $C_{1}$ 
that is computed by Thomason's lollipop algorithm when starting from the cycle $C_{0}$. 
For our proof we simulate the execution of the lollipop algorithm by simultaneously 
assigning to every edge one of four distinct colors in a specific way such that 
four coloring invariants are maintained. Using this coloring procedure, an alternating red-blue 
path is maintained during the execution of the algorithm, which becomes an alternating red-blue 
\emph{cycle} at the end of the execution. As it turns out, this alternating cycle coincides with the 
symmetric difference $C_{0}\ \Delta \ C_{1}$.

This fundamental structural property of the lollipop algorithm (see Theorem~\ref{invariant-4-thm} in Section~\ref{connected-2-factor-sec}) 
has never been revealed so far, and it enables
us to design a novel and more efficient algorithm for detecting a second Hamiltonian cycle of $G$. 
This improves the current state of the art in the computational complexity of \textsc{Smith} among deterministic algorithms running in polynomial space 
(see Section~\ref{alt-cycles-enum-sec}).
Instead of trying to generate the second Hamiltonian cycle $C_1$ directly from $C_0$ (as Thomason's lollipop algorithm does), 
our new algorithm enumerates~--almost--~all alternating red-blue cycles, 
until it finds one alternating cycle $D$ such that 
the symmetric difference $C_{0}\ \Delta \ D$ is a Hamiltonian cycle of $G$ (and not just a collection of cycles that collectively contain all vertices of $G$). 
During its execution, this algorithm iteratively has a choice between two different options for the next edge to be colored red, 
in which cases it branches to create two new instances. 
However, in order for the algorithm to achieve a strictly better worst-case running time than $O^{*}(2^{0.3 n})$, 
it has to refrain from just always blindly branching to new instances. 
We are able to do this by identifying appropriate disjoint quadruples of edges, which we call \emph{ambivalent quadruples}, 
and by \emph{deferring} the choice for the colors of each of these quadruples until the very end. 
Then, at the last step of the algorithm we are able to choose their colors in linear time. 
That is, using the ambivalent quadruples we do not generate \emph{all} possible alternating red-blue cycles but only a succinct representation of them. 
The running time of the algorithm that we eventually achieve is 
$O(n\cdot 2^{0.299862744n}) = O(1.23103^n)$, while our algorithms runs in linear space. 
In the particular case where the input graph $G$ contains no induced cycle~$C_6$ on 6~vertices, the running time becomes $O(n\cdot 2^{0.2971925n}) = O(1.22876^n)$.

In the direction of approximating the length of a second cycle on graphs with minimum degree $\delta$ and maximum degree $\Delta$, 
we provide in Section~\ref{long-cycle-sec} a linear-time algorithm for computing a cycle $C' \neq C_0$ of length at least $n-4a(\sqrt{n}+2\alpha)+8$, where $\alpha = \frac{\Delta-2}{\delta-2}$. 
On the one hand, this improves the results of~\cite{bazgan1999approximation} in two ways. 
First, it provides a direct generalization to arbitrary Hamiltonian graphs of degree at least $3$. 
Second, our algorithm works in linear time in $n$ for all constant-degree regular graphs; in particular it works in time $O(n)$ on cubic graphs (see \cref{cor:regular}). 
On the other hand, we complement the results of~\cite{girao2019long} as we provide a \emph{constructive} proof for their result 
in case where the $\Delta$ and $\delta$ are $o(\sqrt{n})$-factor away from each other. Formally, our algorithm constructs in linear time 
another cycle of length $n - o(n)$ whenever $\frac{\Delta}{\delta} = o(\sqrt{n})$ (see \cref{cor:squareRation}).

\section{Preliminaries\label{preliminaries-sec}}

Given a graph $G=(V,E)$, an edge between two vertices $u$ and $v$ is denoted 
by $uv\in E$, and in this case $u$ and $v$ are said to be \emph{adjacent} in $G$. The \emph{neighborhood} of a vertex $v\in V$ is the set 
$N(v) = \{u\in V : uv \in E\}$ of its adjacent vertices. 
A graph $G$ is cubic if $|N(v)|=3$ for every vertex $v\in V$. 
Given a path $P=(v_1,v_2,\ldots,v_k)$  (resp.~a cycle $C=(v_1,v_2,\ldots,v_k,v_1)$) of $G$, the \emph{length} of $P$ (resp.~$C$) 
is the number of its edges. 
Furthermore, $E(P)$ (resp.~$E(C)$) denotes the set of edges of the path $P$ (resp.~of the cycle $C$).
A path $P$ (resp.~cycle $C$) in~$G$ is a \emph{Hamiltonian} path (resp.~\emph{Hamiltonian} cycle) if it contains each vertex of $G$ exactly once. 
Every cubic Hamiltonian graph is referred to as a \emph{Smith graph}. 
Given a Smith graph $G$ and a Hamiltonian cycle~$C_0$ of~$G$, an edge of $G$ which does not belong to $C_{0}$ is called a \emph{chord} of~$C_{0}$, 
or simply a \emph{chord}. 
The next lemma allows us to assume without loss of generality that the input
Smith graph $G$ is triangle-free (see also Theorem~\ref{triangle-free-thm}).

\begin{lemma}
\label{Smith-triangle-lem}Let $G=(V,E)$ be a Smith graph with $n$ vertices
that contains at least one triangle, and let $C_{0}$ be a Hamiltonian cycle
of $G$. Then either there exists a second Hamiltonian cycle $C_{1}$ of $G$
that contains only two different edges than $C$, or there exists a Smith
graph $G^{\prime }$ with $n-2$ vertices such that every Hamiltonian cycle in 
$G$ corresponds bijectively to a Hamiltonian cycle in $G^{\prime }$ (or
both).
\end{lemma}

\begin{proof}
Let $C_{0}=(v_{1},v_{2},\ldots ,v_{n})$. As $G$ contains at least one
triangle, there must exist a vertex $v_{i}$ such that $v_{i-1}v_{i+1}\in E$,
i.e.~the vertices $v_{i-1},v_{i},v_{i+1}$ build a triangle. 
For the purposes of the proof, denote by $v_{i}^{\ast }$ the unique vertex of $G$ 
which is connected to $v_{i}$ with a chord. If $v_{i}v_{i+2}\in E$, then there exists a 
second Hamiltonian cycle $C_{1}$ of $%
G$, in which the edges $v_{i-1}v_{i},v_{i}v_{i+1},v_{i+1}v_{i+2}$ are
replaced by the edges $v_{i-1}v_{i+1},v_{i+1}v_{i},v_{i}v_{i+2}$. Similarly,
if $v_{i}v_{i-2}\in E$, then there exists a second Hamiltonian cycle $C_{1}$
of $G$, in which the edges $v_{i-2}v_{i-1},v_{i-1}v_{i},v_{i}v_{i+1}$ are
replaced by the edges $v_{i-2}v_{i},v_{i}v_{i-1},v_{i-1}v_{i+1}$. In both
cases, the second Hamiltonian cycle $C_{1}$ of $G$ contains only two
different edges than $C$.

Now suppose that $v_{i}v_{i+2},v_{i}v_{i-2}\notin E$. We prove the statement of
the lemma for the Smith graph $G^{\prime }$ that is obtained from $G$ by
removing vertices $v_{i-1}$ and $v_{i+1}$ and by connecting $v_{i}$ to $%
v_{i-2}$ and $v_{i+2}$. Let $C$ be an arbitrary Hamiltonian cycle of $G$. If 
$C$ contains the edges $v_{i-1}v_{i-2}$ and $v_{i-1}v_{i}$, then it is easy
to see that $C$ must also contain the edges $v_{i}v_{i+1}$ and $%
v_{i+1}v_{i+2}$, and thus $C$ corresponds to a Hamiltonian cycle $C^{\prime
} $ of $G^{\prime }$ in a straightforward way. Otherwise, if $C$ contains
the edges $v_{i-1}v_{i-2}$ and $v_{i-1}v_{i+1}$, then it is easy to see that 
$C$ must also contain the edges $v_{i}v_{i+1}$ and $v_{i}v_{i}^{\ast }$, and
thus in this case $C$ also corresponds to a Hamiltonian cycle $C^{\prime }$
of $G^{\prime }$ in a straightforward way. The third case, where $C$
contains the edges $v_{v-1}v_{i}$ and $v_{i-1}v_{i+1}$ is symmetric to the
previous case. Thus every Hamiltonian cycle $C$ of $G$ corresponds to one
Hamiltonian cycle $C^{\prime }$ of $G^{\prime }$.

Conversely, let $C^{\prime }$ be a Hamiltonian cycle of $G^{\prime }$. If $%
C^{\prime }$ contains the edges $v_{i-2}v_{i}$ and $v_{i}v_{i+2}$, then we
can create a Hamiltonian cycle $C$ of $G$ by replacing these edges with the
edges $v_{i-2}v_{i-1},v_{i-1}v_{i},v_{i}v_{i+1},v_{i+1}v_{i+2}$. If $%
C^{\prime }$ contains the edges $v_{i-2}v_{i}$ and $v_{i}v_{i}^{\ast }$,
then we can create a Hamiltonian cycle $C$ of $G$ by replacing these edges
with the edges $v_{i-2}v_{i-1},v_{i-1}v_{i+1},v_{i}v_{i+1},v_{i}v_{i}^{\ast
} $. The third case where $C^{\prime }$ contains the edges $v_{i}v_{i+2}$
and $v_{i}v_{i}^{\ast }$ is symmetric to the previous case. Thus every
Hamiltonian cycle $C^{\prime }$ of $G^{\prime }$ corresponds to one
Hamiltonian cycle $C$ of $G$.
\end{proof}

\medskip

The next theorem now follows immediately by repeatedly applying Lemma~\ref%
{Smith-triangle-lem}.

\begin{theorem}
\label{triangle-free-thm}Let $G=(V,E)$ be a Smith graph with $n$ vertices
that contains at least one triangle, and let~$C_{0}$ be a Hamiltonian cycle
of~$G$. In linear time we can compute either a second Hamiltonian cycle $%
C_{1}$ of $G$ or a triangle-free Smith graph~$G^{\prime }$ with fewer vertices such
that every Hamiltonian cycle in $G$ bijectively corresponds to a Hamiltonian
cycle in $G^{\prime }$.
\end{theorem}

Now we define the auxiliary notion of an \emph{$X$-certificate} which is a pair of chords forming the shape of an ``$X$'' in a given Hamiltonian cycle. If an $X$-certificate exists then a second Hamiltonian cycle can be immediately computed.

\begin{definition}
\label{x-certificate-def}
Let $G=(V,E)$ be a Smith graph with $n$ vertices and let~$C_{0}=(v_1,v_2,\ldots,v_n)$ be 
a given Hamiltonian cycle of~$G$. 
Let $i,k\in\{1,2,\ldots,n\}$, where $k\notin\{i-1,i,i+1\}$ (here we consider all indices modulo $n$), such that $v_{i}v_{k},v_{i+1}v_{k+1}\in E$. Then the pair $\{v_{i}v_{k},v_{i+1}v_{k+1}\}$ of chords is an \emph{$X$-certificate} of $G$.
\end{definition}

\begin{observation}
\label{x-certificate-obs}
Let $G$ be a Smith graph with $n$ vertices, 
let~$C_{0}=(v_1,v_2,\ldots,v_n)$ be a Hamiltonian cycle of~$G$, 
and let the pair $\{v_{i}v_{k},v_{i+1}v_{k+1}\}$ of chords be an \emph{$X$-certificate} of~$G$, where $i<k$
Then $C_{1}=(v_1,v_2,\ldots,v_i,v_k,v_{k-1},\ldots,v_{i+1},v_{k+1},v_{k+2},\ldots,v_n)$ is a second Hamiltonian cycle of $G$.
\end{observation}

\section{A connected symmetric difference of the two Hamiltonian cycles 
\label{connected-2-factor-sec}}

In this section we present the fundamental structural property of Thomason's lollipop
algorithm that the symmetric difference of the two involved Hamiltonian
cycles is connected. For the sake of presentation, in this section we
simulate Thomason's lollipop algorithm~\cite{T78} on an arbitrary given Smith
graph $G$ and, during this simulation, we assign colors to some of the edges
of $G$. In particular, we assign to some edges of $G$ one of the colors 
\emph{red}, \emph{blue}, \emph{black}, and \emph{yellow}. Note that the
colors of the edges change in every step of the lollipop algorithm.
Furthermore, every such (partial) edge-coloring of $G$ uniquely determines
one step of the lollipop algorithm on $G$ that starts at a specific initial
configuration.

Thomason's lollipop algorithm starts (at Step 0) with a Hamiltonian cycle $%
C_{0}=(v_{1},v_{2},\ldots ,v_{n},v_{1})$; at this step we color all $n$
edges of $C_{0}$ \emph{black}, while all other edges are colored \emph{yellow}.
Any Step $i\geq 1$ of the lollipop algorithm is called \emph{non-final} if
the Hamiltonian path at this step does not correspond to a Hamiltonian
cycle, i.e.~$v_{1}$ is not connected in $G$ to the last vertex of this
Hamiltonian path. 

Step 1 is derived from Step 0 by \emph{removing} the edge $v_{1}v_{n}$ from
the cycle $C_{0}$, thus obtaining the Hamiltonian path $P_{1}=(v_{1},v_{2},%
\ldots ,v_{n})$. We color this removed edge $v_{1}v_{n}$ \emph{red}. Let $%
N(v_{n})=\{v_{1},v_{n-1},v_{k}\}$. At Step 2, the lollipop algorithm
continues by \emph{adding} to the current Hamiltonian path $P_{1}$ the edge $%
v_{n}v_{k}$, thus obtaining a ``lollipop'' in which $v_{k}$ keeps all its
three incident edges, $v_{1}$ keeps only the incident edge $v_{1}v_{2}$, and
every other vertex keeps exactly two of its incident edges. Step 2 is
completed by \emph{removing} the edge $v_{k}v_{k+1}$ from $P_{1}$, thus
``breaking'' the lollipop and obtaining the next Hamiltonian path $%
P_{2}=(v_{1},v_{2},\ldots ,v_{k},v_{n},\ldots ,v_{k+1})$. It is important to
note here that $v_{k+1}$ is the vertex \emph{immediately after} vertex $%
v_{k} $ in the path $P_{i-1}$, where we consider that the path starts at $%
v_{1}$. At Step 2 we color the newly added edge $v_{n}v_{k}$ \emph{blue} and
the removed edge $v_{k}v_{k+1}$ \emph{red}, while the last vertex of the
path $P_{2}$ is $v_{k+1}$. The algorithm continues towards Step 3 by adding
to $P_{2}$ the third edge incident to $v_{k+1}$ (i.e.~the unique incident
edge $v_{k+1}v_{\ell }$ different from the edges $v_{k}v_{k+1}$ and $%
v_{k+1}v_{k+2} $ that belonged to the previous path $P_{1}$) and by removing
again the other incident edge of $v_{\ell }$ that ``breaks'' the lollipop.
Similarly to Step 2, in Step 3 we color the newly added edge $v_{k+1}v_{\ell
}$ \emph{blue} and the newly removed incident edge of $v_{\ell }$ \emph{red}.

As the lollipop algorithm progresses, the (partial) coloring of the edges of 
$G$ continues, according to the following rules at Step $i\geq 1$. Recall
that the Hamiltonian path at Step $i\geq 1$ is denoted by $P_{i}$.
Furthermore, assume that during Step $i$, the path $P_{i}$ is obtained by 
\emph{adding} to $P_{i-1}$ the edge $v_{x}v_{y}$ (where $v_{x}$ is the last
vertex of $P_{i-1}$, thus building a lollipop) and by subsequently \emph{%
removing} from $P_{i-1}$ the edge $v_{y}v_{z}$, thus breaking the
constructed lollipop.

The description of the edge-coloring procedure that we apply at every step
of the lollipop algorithm is formally given by the four coloring rules
below. The intuitive description of these coloring rules is as follows. At
every step, the black edges are those edges of the initial cycle $C_{0}$
which are still contained in the current Hamiltonian path, while the red
edges are all the remaining edges of $C_{0}$, i.e.~those edges which do not
belong to the current Hamiltonian path. The blue edges are those \emph{chords%
} of $C_{0}$ that belong to the current Hamiltonian path. Finally, the
yellow edges are all the remaining chords of $C_{0}$, i.e.~those chords that
do not belong to the current Hamiltonian path. Initially we start with the
cycle $C_{0}$ that contains $n$ black edges and we remove one of them (the
edge $v_{1}v_{n}$) which becomes red. At every step of the algorithm we
build the new \emph{lollipop} when all three incident edges of some vertex $%
v_{y}$ become either black or blue. This can happen either by adding a new
(previously yellow) chord (thus coloring it blue) or by adding a new
(previously colored red) $C_{0}$-edge (thus coloring it black). Once we have
build the new lollipop, we break it within the same step of the lollipop
algorithm, either by removing a (previously colored black) $C_{0}$-edge
(thus coloring it red) or by removing a (previously colored blue) chord
(thus coloring it yellow).

\begin{col-rule}
\label{col-rule-1}If $v_{x}v_{y}\in C_{0}$ then we color $v_{x}v_{y}$ \emph{%
black} (in this case $v_{x}v_{y}$ must be colored red at Step $i-1$).
\end{col-rule}

\begin{col-rule}
\label{col-rule-2}If $v_{x}v_{y}\notin C_{0}$ then we color $v_{x}v_{y}$ 
\emph{blue} (in this case $v_{x}v_{y}$ must be yellow at Step $i-1$).
\end{col-rule}

\begin{col-rule}
\label{col-rule-3}If $v_{y}v_{z}\in C_{0}$ then we color $v_{y}v_{z}$ \emph{%
red} (in this case $v_{y}v_{z}$ must be black at Step $i-1$).
\end{col-rule}

\begin{col-rule}
\label{col-rule-4}If $v_{y}v_{z}\notin C_{0}$ then we color $v_{y}v_{z}$ 
\emph{yellow} (in this case $v_{y}v_{z}$ must be colored blue at Step~$i-1$).
\end{col-rule}

As we prove, the coloring of the edges proceeds such that
the following invariants are maintained:

\begin{invariant}
\label{invar-1}During every non-final Step $i\geq 1$, every $C_{0}$-edge is
colored either \emph{black} or \emph{red}.
\end{invariant}

\begin{invariant}
\label{invar-2}During every non-final Step $i\geq 1$, every non-$C_{0}$-edge
is either \emph{blue} or \emph{yellow}.
\end{invariant}

\begin{invariant}
\label{invar-3}At the end of every non-final Step $i\geq 1$, the set of all
black and blue edges form a Hamiltonian path of~$G$.
\end{invariant}

\begin{invariant}
\label{invar-4}When the lollipop is built during any non-final Step $i\geq 2$%
, the set of all red and blue edges form an \emph{alternating path} of even
length in $G$, starting at $v_{1}$ with a red edge. Furthermore, at the
final step (i.e.~when we build a second Hamiltonian cycle instead of a
lollipop) the set of all red and blue edges form an \emph{alternating cycle} $D$ 
in $G$.
\end{invariant}

The next observations follow immediately from the Coloring Rules~\ref%
{col-rule-1}-\ref{col-rule-4} and by the proof of correctness of the
Thomason's lollipop algorithm~\cite{T78}.

\begin{observation}
\label{invariants-1-2-obs}Invariants~\ref{invar-1} and~\ref{invar-2} are
maintained at every non-final Step $i\geq 1$ of Thomason's lollipop
algorithm.
\end{observation}

\begin{observation}
\label{invariant-3-obs}At every non-final Step $i\geq 1$, the set of black
and blue edges are exactly the edges of the Hamiltonian path at this step of
Thomason's lollipop algorithm, and thus Invariant~\ref{invar-3} is
maintained.
\end{observation}

In the next theorem we prove the maintenance of Invariant~\ref{invar-4}, which is our main technical contribution in this section.

\begin{theorem}
\label{invariant-4-thm}Invariant~\ref{invar-4} is maintained at every (final
or non-final) Step $i\geq 1$ of Thomason's lollipop algorithm. 
Thus, after the final step of the algorithm, the symmetric difference $C_{0} \ \Delta \ C_{1}$ 
of $C_0$ with the produced Hamiltonian cycle $C_1$ is the alternating red-blue cycle $D$.
\end{theorem}

\begin{proof}
The proof is done by induction on the number of steps of the lollipop
algorithm. Invariant~\ref{invar-4} is clearly true at Step 2 of the
algorithm. In fact, when the lollipop is built during Step 2, there is only
one red edge (i.e.~the edge $v_{1}v_{n}$) and one blue edge that is incident
to $v_{n}$, thus creating together a red-blue alternating path of length 2.
In the induction hypothesis, assume that Invariant~\ref{invar-4} is true
until the non-final Step $i\geq 2$. Let $v_{x}$ be the last vertex of the
Hamiltonian path $P_{i-1}$ and let $v_{x}v_{y}$ be the newly added edge that
creates the lollipop during Step $i$ (recall that $v_{x}v_{y}$ is either
blue or black). Furthermore, let $v_{y}v_{z}$ be the edge of the path $%
P_{i-1}$ that is removed in order to break the lollipop, i.e.~$v_{y}v_{z}$
is either a previously black $C_{0}$-edge that is colored red, or a
previously blue chord that is yellow. Recall here that, by construction of
the lollipop algorithm, $v_{z}$ is the vertex \emph{immediately after}
vertex $v_{y}$ in the path $P_{i-1}$, where the path starts at $v_{1}$. Then 
$v_{z}$ becomes the last vertex of the Hamiltonian path $P_{i}$ at Step $i$. Finally, let $v_{z}v_{w}$ be the the newly added (blue or black) edge
that creates the lollipop during Step $i+1$. Note that $v_{w}=v_{1}$ if and
only if Step $i+1$ is the final step.

\medskip

\emph{Case 1.} $v_{x}v_{y}$ is a newly added \emph{blue} edge at Step $i$,
i.e.~a blue chord. Then, by the induction hypothesis, none of the other two
incident edges of $v_{y}$ is blue at Step $i$, and thus they are both black
as they both belong to the path $P_{i-1}$. Therefore $v_{y}$ is the last
vertex of the red-blue alternating path at the time that the lollipop is
built at Step $i$. Furthermore, the edge $v_{y}v_{z}$ is colored \emph{red}
when the lollipop is broken at Step $i$, according to Coloring Rule~\ref%
{col-rule-3}. Now we distinguish the two cases for the color of the newly
added edge $v_{z}v_{w}$ during Step $i+1$.

\medskip

\emph{Case 1a.} $v_{z}v_{w}$ is colored \emph{blue} during Step $i+1$, i.e.~$%
v_{z}v_{w}$ is a yellow chord at Step $i$. First let $v_{w}\neq v_{1}$,
i.e.~Step $i+1$ is not a final step. Then, by the induction hypothesis, none
of the other two incident edges of $v_{w}$ is blue, and thus they are both
black as they both belong to the path $P_{i}$. In this case, the alternating
red-blue path at Step $i$ (which ends at the vertex $v_{y}$ with the blue
edge $v_{x}v_{y}$) is extended by the red edge $v_{y}v_{z}$ and the blue
edge $v_{z}v_{w}$. Now let $v_{w}=v_{1}$, i.e.~Step $i+1$ is the final step.
Then, since until this step vertex $v_{1}$ is an endpoint of the red-blue
alternating path, the addition of the blue edge $v_{z}v_{w}=v_{z}v_{1}$ at
the final Step $i+1$, thus forming an alternating cycle (containing vertex $%
v_{1}$ and the initial red edge $v_{1}v_{n}$). This proves the induction
step in Case 1a, see Figure~\ref{lollipop-alter-cycle-example-steps-fig-1}.

\medskip

\emph{Case 1b.} $v_{z}v_{w}$ is colored \emph{black} during Step $i+1$, i.e.~%
$v_{z}v_{w}$ is a red $C_{0}$-edge at Step $i$. Then, since $v_{z}$ is the
vertex \emph{immediately after} vertex $v_{y}$ in the path $P_{i-1}$, it
follows that $v_{z}\neq v_{1}$. First suppose that $v_{w}=v_{1}$, i.e.~that
Step $i+1$ is the final step. Note that in this case $v_{z}=v_{n}$, since $%
v_{1}v_{n}$ is the only red edge incident to $v_{1}$ until Step $i$.
Furthermore, since $v_{y}v_{z}$ is colored \emph{red} when the lollipop is
broken at Step $i$, it follows that in this case $v_{y}=v_{n-1}$. Note that,
by the induction hypothesis, $v_{z}$ is incident to a blue chord at Step $i$%
, since $v_{z}=v_{n}$ is the second vertex of the red-blue alternating path
with even length. Thus, since $v_{y}v_{z}$ becomes red at the end of Step $i$
and $v_{z}v_{w}$ becomes black at the final Step $i+1$, it follows that at
Step $i+1$ the red and blue edges form an alternating cycle (containing the
red edge $v_{y}v_{z}=v_{n-1}v_{n}$ and not containing vertex $v_{1}$), thus
proving the induction step, see the final step in Figure~\ref%
{lollipop-alter-cycle-example-steps-fig-2}.

Now suppose (within Case 1b) that $v_{w}\neq v_{1}$, i.e.~that Step $i+1$ is
non-final. Recall by the induction hypothesis that, until building the
lollipop at Step $i$, the only red edge that is not incident to two blue
edges is the first red edge, i.e.~$v_{1}v_{n}$. Therefore, since $%
v_{z},v_{w}\neq v_{1}$ and $v_{z}v_{w}$ is red at Step $i$, it follows that
both $v_{z}$ and $v_{w}$ are incident to a blue chord at Step $i$. Thus,
since at Step $i$ the edge $v_{y}v_{z}$ changes color from black to red,
while at Step $i+1$ the edge $v_{z}v_{w}$ changes its color from red to
black, it follows by the induction hypothesis that during Step $i+1$ (when
we break the lollipop) the red and blue edges form an alternating path in $G$%
, starting at $v_{1}$ with a red edge and ending at $v_{w}$ with a blue
edge. This proves the induction step.

\medskip

\emph{Case 2.} $v_{x}v_{y}$ is a newly added \emph{black} edge at Step $i$,
i.e.~a $C_{0}$-edge that was previously colored red. Note that $%
v_{x},v_{y}\neq v_{1}$, since Step $i$ is not the final step by the
induction hypothesis. Denote the three neighbors of $v_{y}$ in $G$ by $v_{x}$%
, $v_{z}$, and $v_{q}$, where $v_{z}$ is the vertex \emph{immediately after}
vertex $v_{y}$ in the path $P_{i-1}$. Since $v_{x}$ is the last vertex of $%
P_{i-1}$ and $v_{x}v_{y}$ is red at the end of Step $i-1$, it follows that
one of the edges $v_{y}v_{z}$ and $v_{y}v_{q}$ is black and the other one is
blue at the end of Step $i-1$. Note that both edges $v_{y}v_{z}$ and $%
v_{y}v_{q}$ maintain their color after building the lollipop at Step $i$.
Thus, since $v_{y}$ has no red incident edge after building the lollipop at
Step $i$, it follows by the induction hypothesis that vertex $v_{y}$ is the
last vertex of the red-blue alternating path.

\medskip

\emph{Case 2a.} $v_{y}v_{z}$ is \emph{blue} (and $v_{y}v_{q}$ is black) at
the beginning of Step $i$. Thus $v_{y}v_{z}$ becomes yellow when we break
the lollipop at Step $i$. Furthermore, since $v_{y}v_{z}$ is a chord, the
edge $v_{z}v_{w}$ is a $C_{0}$-edge and it changes its color from red to
black when we build the lollipop at Step $i+1$. That is, the alternating
red-blue path at Step $i$ (which ends at $v_{z}$ with the blue edge $%
v_{y}v_{z}$) shrinks by removing from it the blue edge $v_{y}v_{z}$ and the
red edge $v_{z}v_{w}$.

Finally note that, if $v_{w}=v_{1}$, then we are left with no red and no
blue edges after building the lollipop at Step $i+1$, while Step $i+1$ is
the final step. However, the absence of red and blue edges at the final step
implies that the obtained Hamiltonian cycle of the lollipop algorithm is the
same as the initial Hamiltonian cycle $C_{0}$, which is a contradiction to
the correctness of the lollipop algorithm~\cite{T78}. Therefore $v_{w}\neq v_{1}$%
, and thus Step $i+1$ is a non-final step. This completes the proof of the
induction step in Case 2a.



\begin{figure}[h!]
\centering%
\subfigure[]{ \label{lollipop-alter-cycle-example-steps-fig-1}
\includegraphics[width=0.747\textwidth]{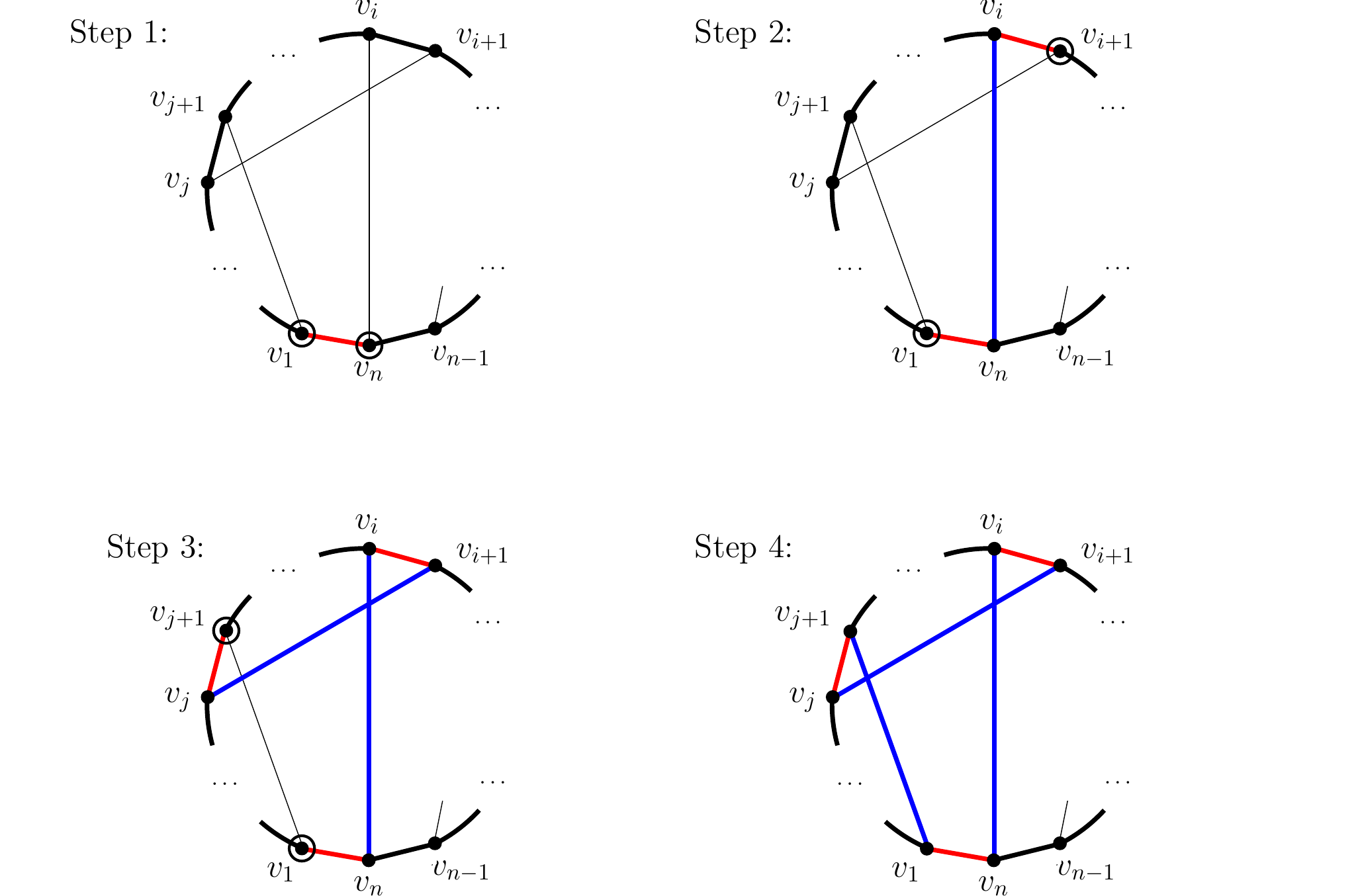}}
\subfigure[]{ \label{lollipop-alter-cycle-example-steps-fig-2}
\includegraphics[width=0.747\textwidth]{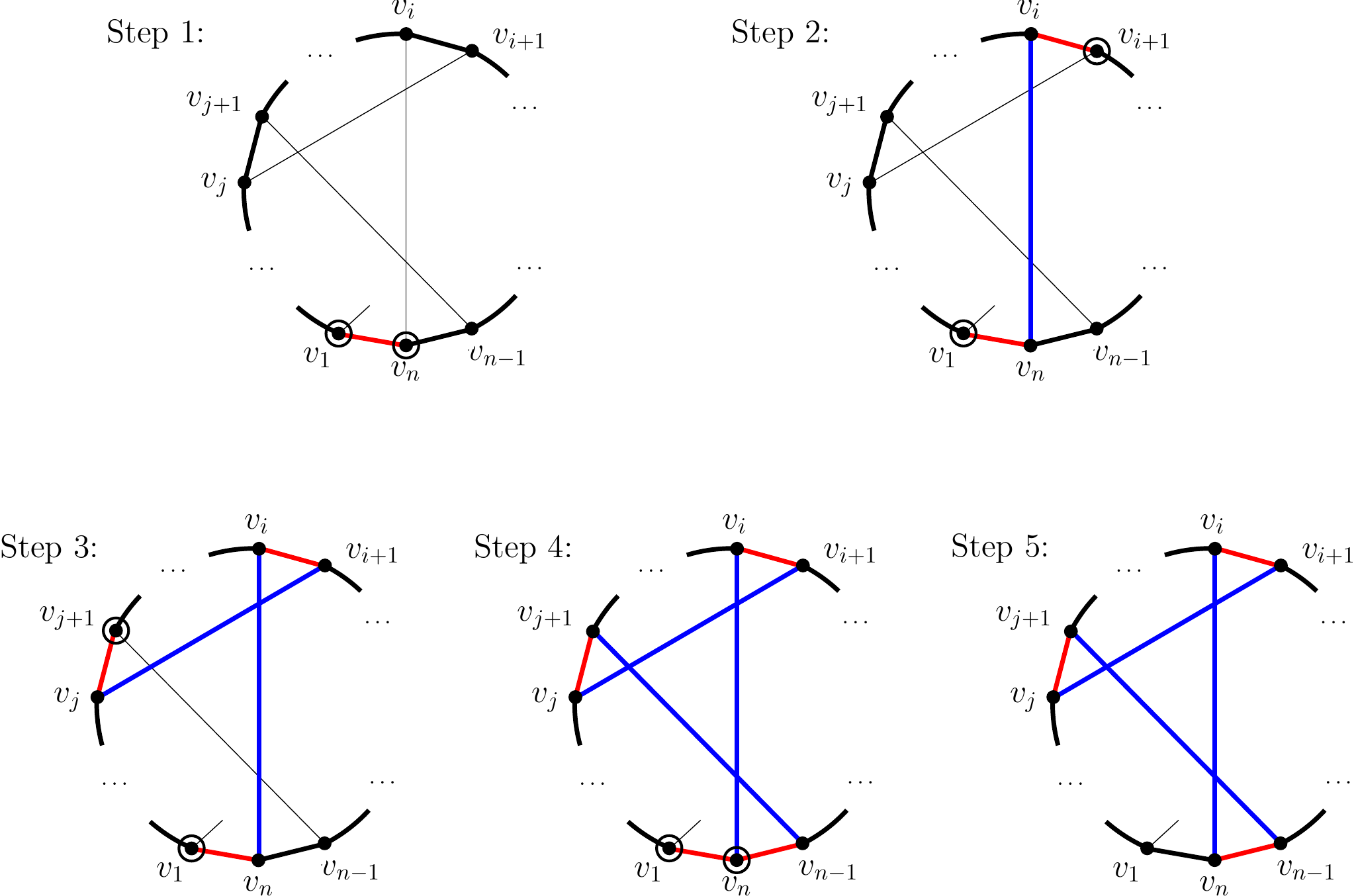}}
\caption{The edge-coloring during the execution of the lollipop algorithm in
two example cubic Hamiltonian graphs. Every non-final Step $i\geq 2$ of the
algorithm encompasses both building the new lollipop (with the blue edge)
and breaking it (with the red edge), thus the illustrated red and blue edges
always have more one red edge (i.e.~the last red edge) than the alternating
red-blue path of even length (see Theorem~\protect\ref{invariant-4-thm}). At
every non-final Step $i$ of the algorithm, the endpoints of the
corresponding Hamiltonian path $P_{i}$ are illustrated by a circled vertex.
In the example~(a), vertex $v_{1}$ belongs to the alternating red-blue
cycle, while in the example~(b) vertex $v_{1}$ does not belong to it (see the
final step in each example).
}
\label{lollipop-alter-cycle-example-steps-fig}
\end{figure}

\emph{Case 2b.} $v_{y}v_{z}$ is \emph{black} (and $v_{y}v_{q}$ is blue)
after building the lollipop at Step $i$, i.e.~the alternating red-blue path
ends at $v_{z}$ with the blue edge $v_{y}v_{q}$. Then the edge $v_{y}v_{z}$
becomes \emph{red} when we break the lollipop at the end of Step $i$.
Furthermore recall that, at the time we build the lollipop at Step $i+1$,
the edge $v_{z}v_{w}$ is either a chord that becomes blue (from yellow) or a 
$C_{0}$-edge that becomes black (from red).

First suppose that $v_{z}v_{w}$ is a chord that becomes \emph{blue} (from
yellow) when we build the lollipop at Step $i+1$. If $v_{w}=v_{1}$ then Step 
$i+1$ is the final step, and in this case the alternating red-blue path
becomes an alternating cycle (containing vertex $v_{1}$ and the initial red
edge $v_{1}v_{n}$) when we build the lollipop at Step $i+1$, which proves
the induction step. Now let $v_{w}\neq v_{1}$, i.e.~Step $i+1$ is a
non-final step. Then, since $v_{z}v_{w}$ is a yellow chord at Step $i$, it
follows by the induction hypothesis that $v_{w}$ is not incident to any red
edge when we build the lollipop at Step $i+1$. Thus, the alternating
red-blue path at Step $i$ (which ends at $v_{y}$ with the blue edge $%
v_{y}v_{q}$) is augmented by adding to it the red edge $v_{y}v_{z}$ and the
blue edge $v_{z}v_{q}$, which again proves the induction step.

Now suppose that $v_{z}v_{w}$ is a $C_{0}$-edge that becomes \emph{black}
(from red) when we build the lollipop at Step $i+1$. Let $v_{w}\neq v_{1}$
(i.e.~Step $i+1$ is non-final). Thus, since at Step $i$ the edge $v_{y}v_{z}$
changes color from black to red, while at Step $i+1$ the edge $v_{z}v_{w}$
changes its color from red to black, it follows by the induction hypothesis
that during Step $i+1$ (when we break the lollipop) the red and blue edges
form an alternating path in $G$, starting at $v_{1}$ with a red edge and
ending at $v_{w}$ with a blue edge. This proves the induction step.

Finally let $v_{w}=v_{1}$, i.e.~Step $i+1$ is the final step. Then $%
v_{z}=v_{n}$, since $v_{1}v_{n}$ is the only red edge incident to $v_{1}$
until Step $i$. Furthermore, since $v_{y}v_{z}$ is colored \emph{red} when
the lollipop is broken at Step $i$, it follows that in this case $%
v_{y}=v_{n-1}$. By the induction hypothesis, $v_{z}$ is incident to a blue
chord at Step $i$, since $v_{z}=v_{n}$ is the second vertex of the red-blue
alternating path with even length. Thus, since $v_{y}v_{z}$ becomes red at
the end of Step $i$ and $v_{z}v_{w}=v_{n}v_{1}$ becomes black at the final
Step $i+1$, it follows that at Step $i+1$ the red and blue edges form an
alternating cycle (containing the red edge $v_{y}v_{z}=v_{n-1}v_{n}$ and not
containing vertex $v_{1}$). This completes the proof of the induction step
in Case 2b.
\end{proof}

\medskip

The next corollary follows by the proof of Theorem~\ref{invariant-4-thm},
and will allow us to reduce the asymptotic running time of our algorithm in
Section~\ref{alt-cycles-enum-sec} by a factor of $n$.

\begin{corollary}
\label{second-cycle-starting-red-edge-cor} Let $C_{0}$ be a given
Hamiltonian cycle of a Smith graph $G$. Let $(v_{i},v_{j},v_{k})$ be three
consecutive vertices of $C_{0}$. Then there exists a second Hamiltonian
cycle $C_{1}$ of $G$ such that (i) $C_{0}\ \Delta \ C_{1}$ is a cycle in $G$
and (ii) either the edge $v_{i}v_{j}$ or the edge $v_{j}v_{k}$ does not
belong to $C_{1}$.
\end{corollary}

\begin{proof}
Part (i) of the corollary follows immediately by the statements of Theorem~\ref{invariant-4-thm} and of Invariant~\ref{invar-4}. To prove part (ii) of
the corollary, first note that, due to symmetry, we may denote without loss
of generality $v_{i}=v_{n-1}$, $v_{j}=v_{n}$, and $v_{k}=v_{1}$.

Within the proof of Theorem~\ref{invariant-4-thm}, there are only two ways
in which the alternating red-blue cycle can be built at the final step of
Thomason's lollipop algorithm. In the first way, the red-blue alternating
cycle contains vertex $v_{1}$ and the red edge $v_{1}v_{n}$ (this can happen
only in Cases 1a and 2b). In the second way, the red-blue alternating cycle
contains the red edge $v_{n-1}v_{n}$ but not vertex $v_{1}$ (this can happen
only in Cases 1b and 2b). This completes the proof of the corollary.
\end{proof}

\section{The alternating cycles' exploration algorithm \label%
{alt-cycles-enum-sec}}

In this section we present our $O(n \cdot 2^{(0.3-\varepsilon)n})$-time algorithm for \textsc{Smith}, 
where $\varepsilon>0$ is a strictly positive constant.
This algorithm improves the
state of the art, as it is asymptotically faster than all known algorithms
for detecting a second Hamiltonian cycle in cubic graphs (among algorithms running in polynomial space). Our algorithm is inspired
by the structural property of Theorem~\ref{invariant-4-thm}. It starts
from a designated vertex $v_{1}$ and constructs an alternating cycle $D$ of
red-blue edges (with respect to $C_{0}$, in the terminology of Section~\ref%
{connected-2-factor-sec}) such that the symmetric difference $C_{0}\
\Delta \ D$ is a Hamiltonian cycle $C_{1}$ of $G$. Equivalently, the
algorithm constructs a second Hamiltonian cycle $C_{1}$ such that the
symmetric difference $D=C_{0}\ \Delta \ C_{1}$ is connected, i.e.~one single
cycle $D$ of $G$ in which every edge alternately belongs to $C_{0}$ and to $%
C_{1}$, respectively.

Before we present and analyze our algorithm (Algorithm~\ref%
{alt-cycle-detection-alg}), we first present some necessary definitions and
notation. Let $G$ be a Smith graph and $C_{0}=(v_{1},v_{2},\ldots ,v_{n})$
be the initial Hamiltonian cycle of $G$. 
For every vertex $v_{i}$ of $G$, we denote by $v_{i}^{\ast }$ the
unique vertex that is connected to $v_{i}$ through a chord. That is,
whenever $v_{i}v_{j}$ is a chord, we have that $v_{j}=v_{i}^{\ast }$ and $%
v_{i}=v_{j}^{\ast }$. Furthermore, every vertex $v_{i}$ is incident to
exactly two $C_{0}$-edges $v_{i-1}v_{i}$ and $v_{i}v_{i+1}$, where we
consider all indices modulo $n$. Algorithm~\ref{alt-cycle-detection-alg}
iteratively \emph{forces} specific edges to be colored \emph{red} ($C_{0}$%
-edges not belonging to $C_{1}$), \emph{black} ($C_{0}$-edges belonging to $%
C_{1}$), \emph{blue} (chords belonging to $C_{1}$), and \emph{yellow}
(chords not belonging to $C_{1}$). Initially, the algorithm starts by
coloring the $C_{0}$-edge $v_{1}v_{t}$ \emph{red}, where $v_{t}\in
\{v_{2},v_{n}\}$, the chord $v_{t}v_{t}^{\ast }$ \emph{blue}, and the two $%
C_{0}$-edges adjacent to the edge $v_{t}v_{1}$ \emph{black}. That is, if $%
v_{t}=v_{2}$ (resp.~if $v_{t}=v_{n}$) then the edges $v_{1}v_{n}$ and $%
v_{2}v_{3}$ (resp.~$v_{1}v_{2}$ and $v_{n-1}v_{n}$) are initially black.
During its execution, the algorithm maintains an alternating red-blue path $%
D $ of \emph{even} length (starting with the red edge $v_{1}v_{t}$ and
ending with a blue edge), until $D$ eventually becomes an alternating cycle.
Note that $D$ can only become a cycle when we color the chord $%
v_{1}v_{1}^{\ast }$ blue. At every iteration the algorithm has (at most) two
choices for the next red edge to be added to $D$, and thus it branches to
(at most) two new instances of the problem, inheriting to both of them the
choices of the forced (i.e.~previously colored) edges made so far. At an
arbitrary non-final step, let $v_{y}$ be the last vertex of the alternating
path $D$, and let $v_{x}v_{y}$ be the last (blue) edge of $D$. For each of
the two $C_{0}$-edges $v_{y-1}v_{y}$ and $v_{y}v_{y+1}$ that are incident to 
$v_{y}$, this edge is called \emph{eligible} if it has not been forced
(i.e.~colored) at a previous iteration; otherwise it is called \emph{%
non-eligible}. Here the term ``eligible''
stands for ``eligible for branching''. We
define the following operations; note that, once an edge has been assigned a
color, it can \emph{never} be forced to change its color.

\vspace{-0,1cm}

\begin{itemize}
\item \textbf{Blue-Branch:} Whenever a chord $v_{x}v_{y}$ is colored blue
(where $v_{y}$ is the last vertex of the current red-blue alternating path $%
D $) and \emph{both} $C_{0}$-edges $v_{y}v_{y+1},v_{y}v_{y-1}$ are eligible,
we create two new instances $I_{1}$ and $I_{2}$, where $I_{1}$ (resp.~$I_{2}$%
) has the edge $v_{y}v_{y+1}$ (resp.~$v_{y}v_{y-1}$) colored red and the
edge $v_{y}v_{y-1}$ (resp.~$v_{y}v_{y+1}$) colored black. \vspace{0cm}

\item \textbf{Blue-Force:} Whenever a chord $v_{x}v_{y}$ is colored blue
(where $v_{y}$ is the last vertex of the current red-blue alternating path $%
D $) and \emph{exactly one} of the two $C_{0}$-edges $%
v_{y}v_{y+1},v_{y}v_{y-1} $ is eligible, we color this eligible $C_{0}$-edge
red.

\item \textbf{Red-Force:} Assume that a $C_{0}$-edge is colored red; note
that this edge must be incident to a blue chord (i.e.~its previous edge in
the alternating path $D$). If its other incident chord is uncolored, we
color it blue. Otherwise, if it has been previously colored yellow, we
announce ``contradiction''. Moreover, if
this new red edge is incident to a $C_{0}$-edge that is uncolored, we color
this edge black.

\item \textbf{Black-Force:} Assume that a $C_{0}$-edge $v_{i}v_{i+1}$ is
colored black, where this edge is adjacent to the (previously colored) black 
$C_{0}$-edge $v_{i-1}v_{i}$ (resp.~$v_{i+1}v_{i+2}$). If their commonly
incident chord $v_{i}v_{i}^{\ast }$ (resp.~$v_{i+1}v_{i+1}^{\ast }$) is so
far uncolored, we color it yellow. \vspace{0cm}Otherwise, if it has been
previously colored blue, we announce ``contradiction''.

\item \textbf{Yellow-Force:} Assume that a chord $v_{i}v_{i}^{\ast }$ is
colored yellow by the operation Black-Force (i.e.~once both $C_{0}$-edges $%
v_{i-1}v_{i},v_{i}v_{i+1}$ become black); furthermore let $v_{k}=v_{i}^{\ast
}$. If at least one of the $C_{0}$-edges $v_{k-1}v_{k},v_{k}v_{k+1}$ has
been previously colored red, we announce ``contradiction''. Otherwise, for each of the $C_{0}$-edges $%
v_{k-1}v_{k},v_{k}v_{k+1}$, if this edge is uncolored, we color it black.
(Note that, if the Yellow-Force operation does not announce
``contradiction'', at the end of the
operation all four $C_{0}$-edges $%
v_{i-1}v_{i},v_{i}v_{i+1},v_{k-1}v_{k},v_{k}v_{k+1}$ that are incident to
the chord $v_{i}v_{i}^{\ast }$ are colored black.)
\end{itemize}

\vspace{-0,1cm}

The main idea of Algorithm~\ref{alt-cycle-detection-alg} is as follows. In
every non-final iteration we have that $D=Red\cup Blue$ is an alternating 
\emph{path}, while in the final iteration $D$ is an alternating \emph{cycle}%
. Suppose that, during a non-final iteration, we extend $D$ by adding a new
blue chord $v_{x}v_{y}$ (where $v_{y}$ is the last vertex of $D$). The cases
where not both edges $v_{y}v_{y+1},v_{y}v_{y-1}$ are eligible are covered by
the following observation.

\begin{observation}
\label{cases-blue-force-obs}The only case, in which at least one of the $%
C_{0}$-edges $v_{y}v_{y+1},v_{y}v_{y-1}$ is red, is when $v_{y}=v_{1}$. In
this case, exactly one of the $C_{0}$-edges $v_{y}v_{y+1},v_{y}v_{y-1}$ is
red and the other one is black (see the initialization lines~\ref{alg-line-1}%
-\ref{alg-line-5} of the algorithm), and thus $D$ becomes an alternating
cycle and the next iteration is the final one. In all other cases, where
none of the $C_{0}$-edges $v_{y}v_{y+1},v_{y}v_{y-1}$ is red, at least one
of them is eligible; otherwise both $v_{y}v_{y+1},v_{y}v_{y-1}$ are black,
and thus their commonly incident chord $v_{x}v_{y}$ has been previously
colored yellow (by the operation Black-Force), a contradiction. If only one
of these two edges is eligible, the algorithm is forced to color this edge
red (with the operation Blue-Force).
\end{observation}

If both edges $v_{y}v_{y+1},v_{y}v_{y-1}$ are eligible, the algorithm \emph{%
branches} (in most cases) to two new instances~$I_{1}$ and $I_{2}$, where $%
I_{1}$ (resp.~$I_{2}$) has the eligible edge $v_{y}v_{y+1}$ (resp.~$%
v_{y}v_{y-1}$) colored red. After the algorithm has branched to these two
new instances $I_{1}$ and $I_{2}$, it exhaustively applies the four forcing
operations Blue-Force, Red-Force, Black-Force, and Yellow-Force, until none
of them is applicable any more. The correctness of these forcing operations
becomes straightforward by recalling our interpretation of the four colors,
i.e.~that the $C_{0}$-edges belonging (resp.~not belonging) to $C_{1}$ are
colored \emph{black} (resp.~\emph{red}), while the chords belonging
(resp.~not belonging) to $C_{1}$ are colored \emph{blue} (resp.~\emph{yellow}).

In some cases, the exhaustive application of the forcing rules in the two
new instances $I_{1},I_{2}$ may only force very few edges, which results in
a large running time of the algorithm before we reach a state where~$D$
becomes an alternating red-blue cycle. To circumvent this problem, we
refrain from just always applying the operation Blue-Branch. Instead, in
some cases we are able to \emph{defer} the choice of the forced color of
specific edges until the very end. More specifically, in some cases we are
able to determine specific sets of~four edges (each containing three $C_{0}$%
-edges and one chord) which build a $C_{4}$ in $G$ (i.e.~a cycle of length 4) 
such that all colored edges in the two different instances $%
I_{1},I_{2}$ are identical, apart from the colors of these four edges.
Therefore all forcing operations in the subsequent iterations of the
algorithm~are \emph{identical} in both these instances $I_{1},I_{2}$,
regardless of the specific colors of these four edges. Furthermore, as it turns
out, every such a quadruple of edges can receive forced colors in \emph{%
exactly two} alternative ways. We call every such a set an \emph{ambivalent
quadruple} of edges. In these few cases, where an ambivalent quadruple
occurs, we do not apply the operation Blue-Branch; instead we continue our
forcing and branching operations in the subsequent iterations of the
algorithm by only starting from one of these instances (instead of starting
from both instances). Then, at the final step of the algorithm, i.e.~when $D$ 
becomes an alternating red-blue cycle, we are able to decide which of the 
two alternative edge colorings is correct for each ambivalent quadruple of edges 
(see the call to Procedure~\ref{ambivalent-proc} in line~\ref{alg-line-call-ambivalent} of the algorithm).

The above crucial trick of not always applying the operation Blue-Branch
allows us to avoid generating \emph{all} possible red-blue alternating
cycles, thus obtaining an exponential speed-up of the algorithm and beating
the state of the art running time of $O^{\ast }(2^{0.3n})$ which is implied
by the TSP-algorithm of~\cite{XiaoNagamochi16}. For example, in one of the
cases where an ambivalent quadruple occurs, if we would branch to two new
instances we would only force 5 new edges. Thus, since $G$ has $\frac{3}{2}n$
edges (as a cubic graph), forcing 5 edges at a time would imply the
generation of at most $O^{\ast }\left( 2^{\frac{3}{2}\cdot \frac{1}{5}%
n}\right) =O^{\ast }\left( 2^{0.3n}\right) $ instances in the worst case,
each of them corresponding to a different red-blue alternating cycle.
However, by deferring the exact coloring of all ambivalent quadruples until
the end of the algorithm, we bypass this problem: instead of generating \emph{%
all possible} red-blue alternating cycles, we create a succinct
representation of them by only generating $O\left(2^{(0.3-\varepsilon)n}\right)$
alternating cycles (for some constant $\varepsilon>0$), and then we determine from them the desired
alternating cycle, i.e.~the one which gives us a second Hamiltonian cycle as
its symmetric difference with the given first Hamiltonian cycle $C_{0}$ 
(see the call to Procedure~\ref{ambivalent-proc} in line~\ref{alg-line-call-ambivalent} of the algorithm). 
Now we define the operation Ambivalent-Flip, which appropriately changes at the end of the algorithm the
already chosen colors of an ambivalent quadruple (see Procedure~\ref%
{ambivalent-proc}). Recall here that every ambivalent quadruple $q$ contains
exactly three $C_{0}$-edges and one chord.

\vspace{-0,1cm}

\begin{itemize}
\item \textbf{Ambivalent-Flip:} Let $q$ be an ambivalent quadruple of
(already colored) edges. For every $C_{0}$-edge of $q$, if it has been
colored red (resp.~black), change its color to black (resp.~red). Also, if
the (unique) chord of $q$ has been colored yellow (resp.~blue), change its
color to blue (resp.~yellow).
\end{itemize}

\vspace{-0,1cm}

Before we proceed with the proof of our main technical lemmas in this
section (see Lemmas~\ref{alt-cycles-correctness-lem} and~\ref%
{alt-cycles-enum-lem}), we first need to define the notions of a \emph{%
forcing path} and a \emph{forcing cycle}. Intuitively, a forcing path
consists of a sequence of edges of $G$ such that, during the execution of
Algorithm~\ref{alt-cycle-detection-alg}, once the first edge is forced to
receive a specific color, every other edge of the path is also forced to
receive some other specific color.

\begin{definition}[forcing path and cycle]
\label{forcing-path-def}Let $G$ be a Smith graph. At an arbitrary iteration
of Algorithm~\ref{alt-cycle-detection-alg}, a path $P=(v_{i_{1}},v_{i_{2}},%
\ldots ,v_{i_{k}})$ of $G$ is a \emph{forcing path starting at vertex $%
v_{i_{1}}$} if:

\begin{itemize}
\item each of its edges $v_{i_{1}}v_{i_{2}},\ldots ,v_{i_{k-1}}v_{i_{k}}$ is
yet uncolored and

\item each of its first $k-1$ vertices $v_{i_{1}},\ldots ,v_{i_{k-1}}$ is
incident to exactly one already colored edge, while its last vertex $%
v_{i_{k}}$ is incident to three yet uncolored edges.
\end{itemize}

Similarly, a cycle $C=(v_{i_{1}},v_{i_{2}},\ldots ,v_{i_{k}},v_{i_{1}})$ of $%
G$ is a \emph{forcing cycle} if:

\begin{itemize}
\item each of its edges $v_{i_{1}}v_{i_{2}},\ldots
,v_{i_{k-1}}v_{i_{k}},v_{i_{k}}v_{i_{1}}$ is yet uncolored and

\item each of its $k$ vertices $v_{i_{1}},\ldots ,v_{i_{k}}$ is incident to
exactly one already colored edge.
\end{itemize}
\end{definition}

Recall that, at every non-final iteration of the algorithm, there is exactly one \emph{%
blue} edge $v_{x}v_{y}$ such that its one endpoint $v_{x}$ is incident to
two other previously colored edges (one red and one black) and its other
endpoint $v_{y}$ is incident either to two uncolored edges or to one
uncolored edge and one black edge. On the other hand, there might be several 
\emph{black} edges $v_{i}v_{j}$ such that $v_{i}$ is incident to two other
previously colored edges and $v_{j}$ is incident to two uncolored edges.
Furthermore, at the end of every iteration of the algorithm, every yellow
and every red edge of $G$ (apart from the first red edge of the alternating
path $D$) is adjacent to four other colored edges. Thus the next observation
follows easily.

\begin{observation}
\label{forcing-path-C0-edges-obs} Let $v_{x}v_{y}$ be the blue chord, where $%
v_{y}$ is the last vertex of the red-blue alternating path $D$ at some
iteration of Algorithm~\ref{alt-cycle-detection-alg}. Furthermore let 
$P=(v_{y},\ldots ,v_{\ell })$ be a forcing path of $G$, starting at vertex~$%
v_{y}$. Then every \emph{internal} vertex of $P$ is incident to one \emph{%
black} $C_{0}$-edge, as well as to one uncolored $C_{0}$-edge and to one
uncolored chord.
\end{observation}

\begin{algorithm}[h!]
	\caption{\textsc{Alternating Cycle Detection}}  
	\label{alt-cycle-detection-alg}  
	\begin{algorithmic}[1]
		\REQUIRE{Instance $I=\{G,C_{0},Q,Red,Blue,Black,Yellow\}$, where $G=(V,E)$ is a Smith graph, 
						  $C_0 = (v_1,v_2,\ldots, v_n)$ is an initial Hamiltonian cycle of $G$, $Q$ is a set of mutually disjoint quadruples of edges, 
						  and $Red,Blue,Black,Yellow$ are four disjoint edge-subsets of $E$ 
						  such that $Red \cup Black \subseteq E(C_0)$ and $Blue \cup Yellow \subseteq E\setminus E(C_0)$.} 
		\ENSURE{A second Hamiltonian cycle $C_1$ of $G$ such that 
$D = C_{0}\ \Delta \ C_{1}$ is connected.}
		
\medskip
		
		\IF[initialization]{$Q=Red=Blue=Black=Yellow=\emptyset$} \label{alg-line-1}
			\IF{there exists an $X$-certificate with the chords $\{v_{i}v_{k},v_{i+1}v_{k+1}\}$}\label{alg-line-X-certif-1}
				\RETURN{the second Hamiltonian cycle $C_{1}=(v_1,v_2,\ldots,v_i,v_k,v_{k-1},\ldots,v_{i+1},v_{k+1},v_{k+2},\ldots,v_n)$}\label{alg-line-X-certif-2}
			\ELSE
				\STATE{Call the algorithm with the parameters:} \label{alg-line-2}
				\COMMENT{Look for an alternating cycle where $v_{1}v_{2}$ is red}
				\STATE{\ \ \ \ $Q \leftarrow \emptyset$; \ \ $Red \leftarrow \{v_{1}v_{2}\}$; \ \ $Blue \leftarrow \{v_{2}v_{2}^{\ast}\}$; \ \ 
                               $Black \leftarrow \{v_{2}v_{3},v_{1}v_{n}\}$; \ \ $Yellow \leftarrow \emptyset$} \label{alg-line-3}
				\medskip
				\STATE{Call the algorithm with the parameters:} \label{alg-line-4}
				\COMMENT{Look for an alternating cycle where $v_{1}v_{n}$ is red}
				\STATE{\ \ \ \ $Q \leftarrow \emptyset$; \ \ $Red \leftarrow \{v_{1}v_{n}\}$; \ \ $Blue \leftarrow \{v_{n}v_{n}^{\ast}\}$; \ \ 
                               $Black \leftarrow \{v_{n-1}v_{n},v_{1}v_{2}\}$; \ \ $Yellow \leftarrow \emptyset$} \label{alg-line-5}
            \ENDIF

        \ELSE[Main Iteration] \label{alg-line-6}
            \STATE{$D \leftarrow Blue \cup Red$} \label{alg-line-7}
            \IF[final iteration of the algorithm]{$D$ is a cycle} \label{alg-line-8}
				\STATE{Call Procedure~\ref{ambivalent-proc}}\label{alg-line-call-ambivalent}


            \ELSE[$D$ is a red-blue alternating path of even length] \label{alg-line-14}
                \STATE{Let $v_{y}$ be the last vertex of 
$D$ (which ends with the blue chord $v_{x}v_{y}$)} \label{alg-line-15}
					\IF[both $v_{y}v_{y+1}$ and $v_{y}v_{y-1}$ are eligible $C_{0}$-edges]{$v_{y}v_{y+1},v_{y}v_{y-1} \notin Black$} \label{alg-line-16}
						
						\STATE{Let $P^{+}$ and $P^{-}$ be the forcing paths starting at $v_{y}$ with the edge $v_{y}v_{y+1}$ and $v_{y}v_{y-1}$, respectively}\label{alg-line-17}%

						\IF{$P^{+} \cup P^{-}$ builds a $C_{4}$ (i.e.~a path with 4 edges)}\label{alg-line-18}
							\STATE{Call Procedure~\ref{C4-proc}}\label{alg-line-19}
						\ELSIF{$P^{+} \cup P^{-}$ builds a $P_{4}$ (i.e.~a path with 4 vertices) whose two endpoints are adjacent}\label{alg-line-20}
							\STATE{Call Procedure~\ref{P4-proc}}\label{alg-line-21}
						\ELSE[$P^{+} \cup P^{-}$ is neither a $C_{4}$ nor a $P_4$ whose two endpoints are adjacent]\label{alg-line-22}
							\STATE{Apply the operation \textbf{Blue-Branch}, generating two new instances $I_1, I_2$}\label{alg-line-23}
							\STATE{Apply exhaustively the operations \textbf{Blue-Force}, \textbf{Red-Force}, \textbf{Black-Force}, \textbf{Yellow-Force}
								   to instances $I_1, I_2$, until no operation can be further applied}\label{alg-line-24}
							\FOR{$i\in \{1,2\}$}\label{alg-line-25}
								\STATE{\textbf{if} no ``contradiction'' has been announced for $I_i$ \textbf{then} call the algorithm on instance $I_i$}\label{alg-line-26}
							\ENDFOR
						\ENDIF
					\ELSE[only one of $v_{y}v_{y+1},v_{y}v_{y-1}$ is an eligible $C_{0}$-edge]\label{alg-line-27}
						\STATE{Apply exhaustively the operations \textbf{Blue-Force}, \textbf{Red-Force}, \textbf{Black-Force}, \textbf{Yellow-Force} 
							   to instance $I$ until no operation can be further applied}\label{alg-line-28}
						\IF{no ``contradiction'' has been announced}\label{alg-line-29}
							\STATE{Call the algorithm on the updated instance $I$}\label{alg-line-30}    
						\ENDIF
					\ENDIF
            \ENDIF
        \ENDIF
	\end{algorithmic}
\end{algorithm}

\begin{algorithm}[t!]
\floatname{algorithm}{Procedure}
\caption{Check the alternating red-blue cycle $D$} \label{ambivalent-proc}
\begin{algorithmic}[1]
\FOR{every ambivalent edge-quadruple $q\in Q$}\label{proc-ambi-line-1}
	\IF{applying \textbf{Ambivalent-Flip} to the colors of $q$ strictly reduces the number of connected components of $C_{0}\ \Delta \ D$}\label{proc-ambi-line-2}
		\STATE{Apply the operation \textbf{Ambivalent-Flip} to $q$ and update $D$ accordingly}\label{proc-ambi-line-3}
	\ENDIF
\ENDFOR
\medskip
\IF{$C_{0}\ \Delta \ D$ is a Hamiltonian cycle of $G$}\label{proc-ambi-line-4}
	\RETURN{Hamiltonian cycle $C_{1}$ and alternating cycle $D$}\label{proc-ambi-line-5}
\ENDIF
\end{algorithmic}
\end{algorithm}

\begin{algorithm}[t!]
\floatname{algorithm}{Procedure}
\caption{Update of instance $I$ when $P^{+} \cup P^{-}$ builds a $C_{4}$} \label{C4-proc}
\begin{algorithmic}[1]
\STATE{$Q \leftarrow Q \cup \{E(P^{+} \cup E(P^{-}))\}$}\label{proc-1-line-1} \COMMENT{new ambivalent quadruple of edges}
\STATE{$Red \leftarrow Red \cup \{v_{y}v_{y+1}\}$; \ \ $Black \leftarrow Black \cup \{v_{y}v_{y-1}\}$}\label{proc-1-line-2}
\STATE{Apply exhaustively the operations \textbf{Blue-Force}, \textbf{Red-Force}, \textbf{Black-Force}, \textbf{Yellow-Force} 
to instance $I$ until no operation can be further applied}\label{proc-1-line-3}
\IF{no ``contradiction'' has been announced}\label{proc-1-line-4}
	\STATE{Call Algorithm~\ref{alt-cycle-detection-alg} on the updated instance $I$}\label{proc-1-line-5}    
\ENDIF
\end{algorithmic}
\end{algorithm}

\begin{algorithm}[t!]
\floatname{algorithm}{Procedure}
\caption{Create instances $I_1,I_2$ when $P^{+} \cup P^{-}$ builds a $P_{4}$ whose two endpoints are adjacent} \label{P4-proc}
\begin{algorithmic}[1]
\STATE{$I_1 \leftarrow I$; \ \ $I_2 \leftarrow I$}\label{proc-2-line-1}
\IF{$P^{+}$ contains only the edge $v_{y}v_{y+1}$ and $P^{-}$ contains the two edges $v_{y}v_{y-1},v_{y-1}v_{y+2}$}\label{proc-2-line-2}
	\STATE{Update $I_1$ such that: \\
		   \ \ \ \ $Red \leftarrow Red \cup \{v_{y}v_{y+1}\}$; \ \ $Q \leftarrow Q \cup \{\{v_{y}v_{y+1}, v_{y}v_{y-1}, v_{y-1}v_{y+2}, v_{y+2}v_{y+1}\}\}$}\label{proc-2-line-3}
	\STATE{Update $I_2$ such that: \\
		   \ \ \ \ $Red \leftarrow Red \cup \{v_{y}v_{y-1}\}$; \ \ $Black \leftarrow Black \cup \{v_{y+1}v_{y+2}\}$}\label{proc-2-line-4}
\medskip
\ELSE[$P^{+}$ contains the two edges $v_{y}v_{y+1},v_{y+1}v_{y-2}$ and $P^{-}$ contains only the edge $v_{y}v_{y-1}$]\label{proc-2-line-5}
	\STATE{Update $I_2$ such that: \\
		   \ \ \ \ $Red \leftarrow Red \cup \{v_{y}v_{y-1}\}$; \ \ $Q \leftarrow Q \cup \{\{v_{y}v_{y-1}, v_{y}v_{y+1}, v_{y+1}v_{y-2}, v_{y-2}v_{y-1}\}\}$}\label{proc-2-line-6}
	\STATE{Update $I_1$ such that: \\
		   \ \ \ \ $Red \leftarrow Red \cup \{v_{y}v_{y+1}\}$; \ \ $Black \leftarrow Black \cup \{v_{y-1}v_{y-2}\}$}\label{proc-2-line-7}
\ENDIF

\medskip

\STATE{Apply exhaustively the operations \textbf{Blue-Force}, \textbf{Red-Force}, \textbf{Black-Force}, \textbf{Yellow-Force} 
to instances $I_1, I_2$, until no operation can be further applied}\label{proc-2-line-8}
\FOR{$i\in \{1,2\}$}\label{proc-2-line-9}
	\STATE{\textbf{if} no ``contradiction'' has been announced for $I_i$ \textbf{then} call Algorithm~\ref{alt-cycle-detection-alg} on instance $I_i$}\label{proc-2-line-10}
\ENDFOR
\end{algorithmic}
\end{algorithm}

In the next lemma (Lemma~\ref{alt-cycles-correctness-lem}) we prove the correctness of our algorithm, 
and after that we prove our crucial technical Lemma~\ref%
{alt-cycles-enum-lem} which specifies how the current instance is
transformed in one iteration of the algorithm. The input instance~$I$ of the
algorithm consists of a Smith graph $G=(V,E)$, a Hamiltonian cycle $C_{0}$
of $G$, the set $Q$ of all ambivalent quadruples, and four disjoint sets of
forced (i.e.~colored) edges $Red$, $Blue$, $Black$, $Yellow$. Initially the
four sets of uncolored edges as well as the set $Q$ are all empty. Given
such an instance $I=(G,C_{0},Q,Red,Blue,Black,Yellow)$, we denote by $%
U(I)=E\setminus \{Red\cup Blue\cup Black\cup Yellow\}$ be the set of all 
\emph{unforced} (i.e.~uncolored) edges in this instance. Furthermore we
denote by $W(I)$ the set of vertices which are not incident to any
edge of $Red\cup Black$ in $I$; we refer to the vertices of $W(I)$ as \emph{%
unbiased} vertices, while all other vertices in $V-W(I)$ are referred to as 
\emph{biased} vertices. Finally, we refer to the set of ambivalent
quadruples of instance $I$ as~$Q(I)$.

\begin{lemma}
\label{alt-cycles-correctness-lem} Let $G=(V,E)$ be a Smith graph and $C_0$
be a Hamiltonian cycle of $G$. Then, Algorithm~\ref{alt-cycle-detection-alg}
correctly computes a second Hamiltonian cycle $C_1$ of $G$ on the input $%
I=(G,C_0,\emptyset,\emptyset,\emptyset,\emptyset,\emptyset)$.
\end{lemma}

\begin{proof}
First recall the interpretation of the four colors: the $C_{0}$-edges that
belong (resp.~do not belong) to the desired Hamiltonian cycle $C_{1}$ are
colored \emph{black} (resp.~\emph{red}), while the chords that belong
(resp.~do not belong) to $C_{1}$ are colored \emph{blue} (resp.~\emph{yellow}). 
In the initialization phase, Algorithm~\ref{alt-cycle-detection-alg} first checks in lines~\ref{alg-line-X-certif-1}-\ref{alg-line-X-certif-2} whether 
an $X$-certificate exists with a pair of chords $\{v_{i}v_{k},v_{i+1}v_{k+1}\}$. 
If such an $X$-certificate exists then the algorithm directly returns the second Hamiltonian cycle 
$C_{1}=(v_1,v_2,\ldots,v_i,v_k,v_{k-1},\ldots,v_{i+1},v_{k+1},v_{k+2},\ldots,v_n)$ and stops its execution; 
this action is correct by Observation~\ref{x-certificate-obs}. 
Otherwise, if $G$ has no $X$-certificate, the algorithm arbitrarily picks vertex $v_{1}$ and it generates two instances $I_{1},I_{2}$
(see lines~\ref{alg-line-2}-\ref{alg-line-5}), where in $I_{1}$ (resp.~$%
I_{2} $) the $C_{0}$-edge $v_{1}v_{2}$ is red, the chord $v_{2}v_{2}^{\ast }$
is blue, and the $C_{0}$-edges $v_{2}v_{3},v_{1}v_{n}$ are black (resp.~$%
v_{1}v_{n}$ is red, $v_{n}v_{n}^{\ast }$ is blue, and $%
v_{n-1}v_{n},v_{1}v_{2}$ are black). Then the algorithm calls itself both on
input $I_{1}$ and on input $I_{2}$. This action of the algorithm is correct
by Corollary~\ref{second-cycle-starting-red-edge-cor}, since the algorithm
searches for the desired second Hamiltonian cycle $C_{1}$ such that either $%
v_{1}v_{2}$ or $v_{1}v_{n}$ is red (i.e.~does not belong to $C_{1}$).

In each of these initial cases, the algorithm starts with one red $C_{0}$%
-edge and one blue chord (i.e.~with an alternating red-blue path $D$ of
length two), and it iteratively extends $D$ by adding one red $C_{0}$-edge
and one blue chord to it, until either a ``contradiction'' is announced (by
one of the forcing rules), or $D$ becomes an alternating red-blue cycle by
hitting the first red edge of $D$ at vertex $v_{1}$ with the blue chord $%
v_{1}v_{1}^{\ast }$. In every non-final iteration of the algorithm,
i.e.~when $D$ is still an alternating path, the algorithm proceeds as
follows. Suppose that $D$ ends at vertex $v_{y}$ with the blue chord $%
v_{x}v_{y}$. Due to Observation~\ref{cases-blue-force-obs}, at most one of
the two incident $C_{0}$-edges $v_{y}v_{y+1},v_{y}v_{y-1}$ is red. If
exactly one of them is red, then the other incident $C_{0}$-edge is black,
while $D$ becomes an alternating cycle and the next iteration is the last
one. Furthermore, it also follows by Observation~\ref{cases-blue-force-obs}
that, whenever none of the $C_{0}$-edges $v_{y}v_{y+1},v_{y}v_{y-1}$ is red,
there is at least one eligible (i.e.~uncolored) edge and at most one black
edge anong $v_{y}v_{y+1},v_{y}v_{y-1}$.

If exactly one of these two edges is eligible and the other one is black
(see lines~\ref{alg-line-27}-\ref{alg-line-30}), then the algorithm is
forced to color this eligible edge red. Thus, in this case the algorithm
correctly updates the current instance $I$ by exhaustively applying the
forcing operations Blue-Force, Red-Force, Black-Force, and Yellow-Force
until none of them can be applied any more. Now suppose that both edges $%
v_{y}v_{y+1},v_{y}v_{y-1}$ are eligible, and let $P^{+}$ (resp.~$P^{-}$) be
the forcing path starting at vertex $v_{y}$ with the edge $v_{y}v_{y+1}$
(resp.~$v_{y}v_{y-1}$). Furthermore suppose that none of the conditions of
lines~\ref{alg-line-18} and~\ref{alg-line-20} are satisfied. Then the
algorithm first branches into two new instances $I_{1},I_{2}$ (by applying
Blue-Force at vertex $v_{y}$) and it then exhaustively applies the four
forcing rules (see lines~\ref{alg-line-22}-\ref{alg-line-26}). The
correctness of these forcing operations becomes straightforward by recalling
our interpretation of the four colors. At the final iteration of the
algorithm (see the call to Procedure~\ref{ambivalent-proc} in line~\ref%
{alg-line-call-ambivalent} of the algorithm), if the set $Q$ of all
ambivalent quadruples is empty, then the algorithm just checks whether the
symmetric difference between $C_{0}$ and the produced alternating red-blue
cycle $D$ is just one cycle. The correctness of this check follows by
Corollary~\ref{second-cycle-starting-red-edge-cor}.

It remains to prove the correctness of the algorithm also in the case where
one of the conditions of lines~\ref{alg-line-18} and~\ref{alg-line-20} is
satisfied. First we analyze each of these two cases separately, as follows.

\medskip

\emph{Case 1: line~\ref{alg-line-18} is applied.} The union $P^{+}\cup P^{-}$
of the two forcing paths builds a $C_{4}$ (i.e.~a cycle with 4 edges).
Assume that each of the paths $P^{+},P^{-}$ has two edges, i.e.~$P^{+}$
contains the $C_{0}$-edge $v_{y}v_{y+1}$ and the chord $v_{y+1}v_{y+1}^{\ast
}$, while $P^{-}$ contains the $C_{0}$-edge $v_{y}v_{y-1}$ and the chord $%
v_{y-1}v_{y-1}^{\ast }$. Then, since $P^{+}\cup P^{-}$ is a $C_{4}$, it
follows that $v_{y+1}^{\ast }=v_{y-1}^{\ast }$, and thus $v_{y+1}^{\ast }$
is incident to two different chords, which is a contradiction as $G$ is a
cubic graph. Therefore, one of the forcing paths $P^{+},P^{-}$ has length 1
and the other one has length 3. As these two cases are symmetric, assume
without loss of generality that $P^{+}$ has length 1 (i.e.~it only contains
the $C_{0}$-edge $v_{y}v_{y+1}$) and $P^{-}$ has length 3. Since vertex $%
v_{y+1}$ is the common endpoint of $P^{+}$ and $P^{-}$, it follows that $%
P^{-}$ contains the $C_{0}$-edge $v_{y}v_{y-1}$, the chord $v_{y-1}v_{y+2}$,
and the $C_{0}$-edge $v_{y+2}v_{y+1}$ (in this order). That is, the cycle $%
P^{+}\cup P^{-}$ contains exactly three $C_{0}$-edges and one chord.
Furthermore note that the third edge incident to $v_{y+1}$ (apart from $%
v_{y} $ and $v_{y+2}$) is the chord $v_{y+1}v_{y+1}^{\ast }$.

Applying the operation Blue-Branch at this iteration would result in the
creation of two new instances $I_{1},I_{2}$, in which the edges are colored
as follows. In $I_{1}$, the $C_{0}$-edge $v_{y}v_{y+1}$ is colored red, the $%
C_{0}$-edges $v_{y}v_{y-1},v_{y+2}v_{y+1}$ are colored black, and the chord $%
v_{y-1}v_{y+2}$ is colored yellow. In $I_{2}$, the $C_{0}$-edge $%
v_{y}v_{y+1} $ is colored black, the $C_{0}$-edges $%
v_{y}v_{y-1},v_{y+2}v_{y+1}$ are colored red, and the chord $v_{y-1}v_{y+2}$
is colored blue. Note that, in \emph{both} instances $I_{1},I_{2}$, these
edge colorings force the chord $v_{y+1}v_{y+1}^{\ast }$ to be colored blue,
and these are all edge colorings that can be forced so far.

That is, the only difference between the instances $I_{1},I_{2}$ is the way
the four edges of $P^{+}\cup P^{-}$ are colored. Therefore, since each of
the vertices of $P^{+}\cup P^{-}$ is ``saturated'' in both $I_{1},I_{2}$
(i.e.~it has all its three incident edges colored), all forcing operations
in the subsequent iterations of Algorithm~\ref{alt-cycle-detection-alg} are 
\emph{identical} in both $I_{1},I_{2}$. Using this fact, the algorithm marks
the edges of $P^{+}\cup P^{-}$ as an \emph{ambivalent} quadruple of edges
(see line~\ref{proc-1-line-1} of Procedure~\ref{C4-proc}). Furthermore, it
colors these 4 edges according to $I_{1}$ only, i.e.~without branching to
both $I_{1},I_{2}$, and it calls itself on the updated instance $I$ (see
lines~\ref{proc-1-line-2}-\ref{proc-1-line-5} of Procedure~\ref{C4-proc}).
Since all subsequent forcing operations would be identical in both instances 
$I_{1},I_{2}$, the algorithm continues either until a contradiction is
concluded at a later iteration (in this case a contradiction would be
concluded by both $I_{1},I_{2}$) or until $D$ becomes an alternating cycle.

\medskip

\emph{Case 2: line~\ref{alg-line-20} is applied.} The union $P^{+}\cup P^{-}$
of the two forcing paths builds a $P_{4}$ (i.e.~a path with 4 vertices)
whose two endpoints are adjacent. In this case, clearly one of the paths $%
P^{+},P^{-}$ contains one edge and the other one contains two edges. As
these two cases are symmetric (see lines~\ref{proc-2-line-2}-\ref%
{proc-2-line-4} and lines~\ref{proc-2-line-5}-\ref{proc-2-line-7} of
Procedure~\ref{P4-proc}, respectively), it suffices to only analyze here the
case that $P^{+}$ contains only the $C_{0}$-edge $v_{y}v_{y+1}$ and $P^{-}$
contains the $C_{0}$-edge $v_{y}v_{y-1}$ and the chord $v_{y-1}v_{y+2}$.
Note that, by the assumption of Case 2, the endpoints of the paths $%
P^{+},P^{-}$ are connected via the $C_{0}$-edge $v_{y+1}v_{y+2}$.
Furthermore, note that the third edge incident to $v_{y+1}$ (apart from $%
v_{y}$ and $v_{y+2}$) is the chord $v_{y+1}v_{y+1}^{\ast }$; similarly, the
third edge incident to $v_{y+2}$ (apart from $v_{y-1}$ and $v_{y+1}$) is the 
$C_{0}$-edge $v_{y+2}v_{y+3}$.

Applying the operation Blue-Branch at this iteration would result in the
creation of two new instances $I_{1},I_{2}$, in which the edges are colored
as follows. In $I_{1}$, the $C_{0}$-edge $v_{y}v_{y+1}$ is colored red, the $%
C_{0}$-edges $v_{y}v_{y-1},v_{y+1}v_{y+2},v_{y+2}v_{y+3}$ are colored black,
the chord $v_{y-1}v_{y+2}$ is colored yellow, and the chord $%
v_{y+1}v_{y+1}^{\ast }$ is colored blue. On the other hand, in $I_{2}$ the $%
C_{0}$-edge $v_{y}v_{y+1}$ is colored black, the $C_{0}$-edge $v_{y}v_{y-1}$
is colored red, and the chord $v_{y-1}v_{y+2}$ is colored blue. Consider now
applying again the operation Blue-Branch in the instance $I_{2}$ at vertex $%
v_{y+2}$ (once the chord $v_{y-1}v_{y+2}$ has been colored blue in $I_{2}$).
This would replace instance $I_{2}$ by two new instances $I_{2}^{1}$ and $%
I_{2}^{2}$, in which the edges are colored as follows. In $I_{2}^{1}$, the $%
C_{0}$-edges $v_{y}v_{y+1},v_{y+2}v_{y+3}$ are colored black, the $C_{0}$%
-edges $v_{y}v_{y-1},v_{y+1}v_{y+2}$ are colored red, and the chords $%
v_{y-1}v_{y+2},v_{y+1}v_{y+1}^{\ast }$ are colored blue. Furthermore, in $%
I_{2}^{2}$ the $C_{0}$-edges $v_{y}v_{y+1},v_{y+1}v_{y+2}$ are colored
black, the $C_{0}$-edges $v_{y}v_{y-1},v_{y+2}v_{y+3}$ are colored red, the
chord $v_{y-1}v_{y+2}$ is colored blue, and the chord $v_{y+1}v_{y+1}^{\ast}$
is colored yellow.

Now note that both instances $I_{1}$ and $I_{2}^{1}$ are identical, apart
from the colors of the four edges $%
v_{y}v_{y+1},v_{y}v_{y-1},v_{y-1}v_{y+2},v_{y+1}v_{y+2}$. Therefore, since
each of these four vertices $v_{y},v_{y+1},v_{y-1},v_{y+1}$ is ``saturated''
in both $I_{1}$ and $I_{2}^{1} $ (i.e.~it has all its three incident edges
colored), all forcing operations in the subsequent iterations of Algorithm~%
\ref{alt-cycle-detection-alg} are \emph{identical} in both $I_{1},I_{2}^{1}$%
. Thus, instead of branching into the three instances $%
I_{1},I_{2}^{1},I_{2}^{2}$, the algorithm only branches into the two
instances $I_{1}$ and $I_{2}^{2}$ in Procedure~\ref{P4-proc}, while it also
marks the above four edges as an \emph{ambivalent} quadruple of edges within 
$I_{1}$ (see line~\ref{proc-2-line-3} of Procedure~\ref{P4-proc}). Then,
within the recursive call on instance $I_{1}$, the algorithm continues
either until a contradiction is concluded at a later iteration (in this case
a contradiction would be concluded by both $I_{1},I_{2}^{1}$) or until $D$
becomes an alternating cycle.

\medskip

\emph{Correctness for both Cases 1 and 2.} Suppose that, at some iteration
of Algorithm~\ref{alt-cycle-detection-alg}, $D$ becomes an alternating
cycle, and let $Q$ be the set of ambivalent edge-quadruples that the
algorithm has marked so far. Then the algorithm calls Procedure~\ref%
{ambivalent-proc} (see line~\ref{alg-line-8} of the algorithm). Recall by
the above analysis of Cases 1 and 2 that, for every ambivalent
edge-quadruple $q\in Q$, the algorithm had to choose between two alternative
edge colorings of the four edges of $q$. Furthermore, until the execution of
lines~\ref{proc-ambi-line-1}-\ref{proc-ambi-line-3} of Procedure~\ref%
{ambivalent-proc}, these choices were made arbitrarily, as the choice
between these two alternative colorings of $q$ had no effect on the
subsequent iterations of the algorithm. Now, performing the operation
Ambivalent-Flip at an ambivalent edge-quadruple $q\in Q$, is equivalent to
choosing the second alternative coloring of $q$.

Assume that, at the beginning of Procedure~\ref{ambivalent-proc}, the
symmetric difference $C_{0}\ \Delta \ D$ has $k$ cycles. In lines~\ref%
{proc-ambi-line-1}-\ref{proc-ambi-line-3} of the procedure, the algorithm
attempts to sequentially perform the operation Ambivalent-Flip on all
ambivalent quadruples $q\in Q$. By flipping the coloring of such a quadruple 
$q$, the number $k$ of connected components of $C_{0}\ \Delta \ D $ can
either reduce to $k-1$, or increase to $k+1$, or stay unchanged at $k$. Note
that, if it decreases to $k-1$ , then flipping the colors of $q$ connects
two different connected components (i.e.~cycles) of $C_{0}\ \Delta \ D$ into
one. Similarly, if it increases to $k+1$, then flipping the colors of $q$
disconnects one cycle of $C_{0}\ \Delta \ D$ into two different ones.
Finally, if it stays unchanged at $k$, then flipping the colors of $q$
simply replaces one cycle of $C_{0}\ \Delta \ D$ with another cycle that
visits the same vertices in a different order.

Now note that, performing the operation Ambivalent-Flip at one
edge-quadruple $q$, does not change the color of an edge in any other
quadruple $q^{\prime }\in Q\setminus \{q\}$. Thus, at the end of the
execution of lines~\ref{proc-ambi-line-2}-\ref{proc-ambi-line-3} of
Procedure~\ref{ambivalent-proc}, the algorithm can decide whether there
exists a sequence of choices for the edge-colorings of the ambivalent
quadruples in $Q$ which derive the desired second Hamiltonian cycle $C_{1}$.
More specifically, if the symmetric difference $C_{0}\ \Delta \ D$ is
connected (see~line~\ref{proc-ambi-line-4} of Procedure~\ref{ambivalent-proc}%
, where now $D$ is the updated red-blue cycle after exhaustively executing
lines~\ref{proc-ambi-line-2}-\ref{proc-ambi-line-3} of the procedure) then $%
C_{0}\ \Delta \ D$ is the desired second Hamiltonian cycle. Otherwise, if $%
C_{0}\ \Delta \ D$ still contains more than one cycle, it follows that no
sequence of choices for the alternative edge-colorings of the quadruples in $%
Q$ could derive a second Hamiltonian cycle. This completes the proof of the lemma.
\end{proof}

\begin{lemma}
\label{alt-cycles-enum-lem}Let $I=(G,C_{0},Q,Red,Blue,Black,Yellow)$ be the
instance at some iteration of Algorithm~\ref{alt-cycle-detection-alg}, where 
$G=(V,E)$ is a Smith graph, and let $D=Red\cup Blue$ be the current
alternating red-blue path of even length. Then, within a constant number of
iterations, either a ``contradiction'' is announced or the algorithm
transforms the instance $I$ in lines~\ref{alg-line-16}-\ref{alg-line-30}
either to a single instance $I^{\prime }$, where $|U(I^{\prime })|\leq
|U(I)|-2$, or to two instances $I_{1}$ and $I_{2} $, where one of the
following is satisfied:

\begin{enumerate}
\item $|W(I_{1})|,|W(I_{2})|\leq |W(I)|-2$ and $|U(I_{1})|,|U(I_{2})|\leq
|U(I)|-7$,\label{recurrence-eq-1a}

\item $|W(I_{1})|,|W(I_{2})|\leq |W(I)|-2$ and $|U(I_{1})|,|U(I_{2})|\leq
|U(I)|-9$,\label{recurrence-eq-1b}

\item $|W(I_{1})|,|W(I_{2})|\leq |W(I)|-4$ and $|U(I_{1})|,|U(I_{2})|\leq
|U(I)|-4$,\label{recurrence-eq-2}

\item $|W(I_{1})|\leq |W(I)|-4$, $|U(I_{1})|\leq |U(I)|-4$, and $%
|W(I_{2})|\leq |W(I)|-4$, $|U(I_{2})|\leq |U(I)|-6$,\label{recurrence-eq-3}

\item $|W(I_{1})|\leq |W(I)|-2$, $|U(I_{1})|\leq |U(I)|-9$, and $%
|W(I_{2})|\leq |W(I)|-4$, $|U(I_{2})|\leq |U(I)|-6$,\label{recurrence-eq-4}

\item $|W(I_{1})|\leq |W(I)|-2$, $|U(I_{1})|\leq |U(I)|-5$, and $%
|W(I_{2})|\leq |W(I)|-4$, $|U(I_{2})|\leq |U(I)|-8$,\label{recurrence-eq-5}

\item $|W(I_{1})|\leq |W(I)|-2$, $|U(I_{1})|\leq |U(I)|-3$, and $%
|W(I_{2})|\leq |W(I)|-6$, $|U(I_{2})|\leq |U(I)|-7$,\label{recurrence-eq-6}

\item $|W(I_{1})|\leq |W(I)|-2$, $|U(I_{1})|\leq |U(I)|-3$, and 
$|W(I_{2})|\leq |W(I)|-4$, $|U(I_{2})|\leq |U(I)|-10$,\label{recurrence-eq-7}

\item $|W(I_{1})|\leq |W(I)|-2$, $|U(I_{1})|\leq |U(I)|-3$, and $%
|W(I_{2})|\leq |W(I)|-5$, $|U(I_{2})|\leq |U(I)|-9$.\label{recurrence-eq-8}
\end{enumerate}
\end{lemma}

\begin{proof}
Let $v_{y}$ be the last vertex of the alternating red-blue path $D$, and let 
$v_{x}v_{y}$ be its last blue chord. Throughout the proof we assume that no
``contradiction'' is announced at the current iteration. Recall that the set 
$W(I)$ contains all unbiased vertices of instance $I$, i.e.~all vertices which
are not incident to any edge of $Red\cup Black$. Furthermore recall that the set $%
U(I)$ contains all unforced edges of instance $I$, i.e.~all edges that are
not contained in the set $Red\cup Blue\cup Black\cup Yellow$. Since $D$ is
an alternating red-blue path (and not a red-blue cycle) by the assumption of
the lemma, note that at least one of the two $C_{0}$-edges $%
v_{y}v_{y+1},v_{y}v_{y-1}$ is eligible and at most one of them is already
colored black, see Observation~\ref{cases-blue-force-obs}. Assume that one
of these two edges is colored and the other one uncolored; note that the
colored one can only be black. In this case the other edge is forced to be
colored red by the operation Blue-Force. Furthermore, this triggers the
operation Red-Force, which forces its uncolored incident chord blue. Thus,
in this case the algorithm reduces the problem to a new single instance $%
I^{\prime }$, which has at least two more edges colored, i.e.~$|U(I^{\prime
})|\leq |U(I)|-2$. This corresponds to the case~(i) of the lemma.

For the remainder of the proof assume that both $C_{0}$-edges $%
v_{y}v_{y+1},v_{y}v_{y-1}$ are eligible. Then the operation Blue-Branch
takes place and creates two instances $I_{1},I_{2}$ (except the cases where
line~\ref{alg-line-19} or line~\ref{alg-line-21} of Algorithm~\ref%
{alt-cycle-detection-alg} is executed, which are dealt with separately in
the proof below), where in $I_{1}$ the edge $v_{y}v_{y+1}$ is red and the
edge $v_{y}v_{y-1}$ is black, and in $I_{2}$ the edge $v_{y}v_{y-1}$ is red
and the edge $v_{y}v_{y+1}$ is black. Note that, in both cases, $v_{y}$
becomes a new biased vertex at this iteration as it becomes incident to both
a red edge and a black edge. It is not hard to see that the two $C_{0}$%
-edges $v_{y}v_{y+1},v_{y}v_{y-1}$ cannot participate together in a forcing
cycle. Indeed, in such a forcing cycle $C$, one of the edge sequences $%
v_{y}v_{y+1},v_{y+1}v_{y+1}^{\ast },\ldots $ and $%
v_{y}v_{y-1},v_{y-1}v_{y-1}^{\ast },\ldots $ altervatively receives the
colors red and blue, while the other one altervatively receives the colors
black and yellow. Therefore, there exists exactly one forcing path $P^{-}$
starting at vertex $v_{y}$ with the edge $v_{y}v_{y-1}$, and exactly one
forcing path $P^{+}$ starting at vertex $v_{y}$ with the edge $v_{y}v_{y+1}$%
. The next observation follows easily from the fact that for every
previously colored edge $v_{i}v_{j}$, at least one of its endpoints $%
v_{i},v_{j}$ is incident to \emph{three} previously colored edges.

\begin{observation}
\label{forced-paths-disjoint-internal-vertices-obs}The two forcing paths $%
P^{-},P^{+}$ do not share any common \emph{internal} vertex.
\end{observation}

\emph{Case 1.} Both $P^{-}$ and $P^{+}$ end with a $C_{0}$-edge. That is,
each of these forcing paths has an even number of \emph{internal} vertices.
Since the analysis for both new instances $I_{1},I_{2}$ is symmetric, in
most subcases of Case 1 (with the exception of Case 1(iv)) we only analyze
instance $I_{1}$, i.e.~the case where $v_{y}v_{y+1}$ becomes red and $%
v_{y}v_{y-1}$ becomes black by the operation Blue-Branch. Let~$v_{\ell }$ be
the last vertex of $P^{+}$, and assume without loss of generality that the
last edge of $P^{+}$ is $v_{\ell -1}v_{\ell }$ (the other case, where the
last edge of $P^{+}$ is $v_{\ell }v_{\ell +1}$ is exactly symmetric).
Similarly, let $v_{q}$ be the last vertex of $P^{-}$, and assume without
loss of generality that the last edge of $P^{+}$ is $v_{q-1}v_{q}$. That is, 
$v_{\ell -1}$ and $v_{q-1}$ are the last \emph{internal} vertices of $P^{+}$
and of $P^{-}$, respectively. Note that $v_{q}$ becomes a new biased vertex
after the forcing operations along $P^{-}$, as it becomes incident to the
new black edge $v_{q-1}v_{q}$. Similarly $v_{\ell }$ becomes a new biased
vertex after the forcing operations along $P^{+}$, as it becomes incident to
the new red edge $v_{\ell -1}v_{\ell }$ and to the new black edge $v_{\ell
}v_{\ell +1}$. That is, $v_{y},v_{q},v_{\ell }$ become new biased vertices.

\emph{Case 1(i).} $v_{\ell }=v_{q}$. Then, since $P^{+}$ and $P^{-}$ share $%
v_{y}$ as a common vertex, $P^{+}\cup P^{-}$ cannot have just two edges,
i.e.~$P^{+}\cup P^{-}$ has at least 4 edges. First assume that $P^{+}\cup
P^{-}$ has 4 edges, that is, $P^{+}\cup P^{-}$ is a $C_{4}$. Then the
algorithm executes line~\ref{alg-line-19} and calls Procedure~\ref{C4-proc}.
In this case, it only updates the current instance $I$ by forcing the colors
of at least 5 edges, namely the 4 edges of $P^{+}\cup P^{-}$ as well as the
chord $v_{\ell }v_{\ell }^{\ast }$. Moreover, the updated instance has at
least 2 new biased vertices $v_{y},v_{\ell }$. Furthermore the 4 edges of $%
P^{+}\cup P^{-}$ are added as an ambivalent quadruple in the set $Q$ (see
line~\ref{proc-1-line-1} of Procedure~\ref{C4-proc}).

Now assume that $P^{+}\cup P^{-}$ has exactly 6 edges. That is, either each
of the paths $P^{+}$ and $P^{-}$ contains three edges, or one of them
contains one edge and the other one contains five edges. Then the algorithm
executes lines~\ref{alg-line-22}-\ref{alg-line-26} and creates two new
instances $I_{1},I_{2}$, each of them having at least 2 new biased vertices $%
v_{y},v_{\ell }$ and at least 7 new forced edges, namely 6 forced edges in $%
P^{+}\cup P^{-}$ as well as the blue chord $v_{\ell }v_{\ell }^{\ast }$.
Note that, in this case, four of the vertices of $P^{+}\cup P^{-} $ (i.e.~all vertices of $P^{+}\cup P^{-}$ apart from $v_{y}$ and $v_{\ell }$) are
incident to one previously colored black $C_{0}$-edge. Furthermore, note
that exactly four $C_{0}$-edges and three chords are being forced
(i.e.~colored). In addition, these four newly forced $C_{0}$-edges are
either two pairs of consecutive $C_{0}$-edges, or three consecutive $C_{0}$%
-edges and one separate $C_{0}$-edge (i.e.~not consecutive with the other
three ones).

Finally, assume that $P^{+}\cup P^{-}$ has at least 8 edges. Then, similarly
to the previous paragraph, we have at least 2 new biased vertices and 9 new
forced edges.

Summarizing, in Case 1(i) we either have two instances $I_{1},I_{2}$, each
having at least 2 new biased vertices and 7 forced edges (when $P^{+}\cup
P^{-}$ has exactly 6 edges), or at least 2 new biased vertices and 9 forced
edges (when $P^{+}\cup P^{-}$ has at least 8 edges), or we just have one
updated instance $I$ which has 2 new biased vertices and 5 forced edges.

\emph{Case 1(ii).} $v_{\ell }=v_{q+1}$ (or equivalently, $v_{q}=v_{\ell +1}$%
). In this case, when edge $v_{\ell -1}v_{\ell }$ (i.e.~the last edge of $%
P^{+}$) is colored red, the operation Red-Force is triggered which colors
the $C_{0}$-edge $v_{\ell }v_{q}$ black and the chord $v_{\ell }v_{\ell
}^{\ast }$ blue. On the other hand, note that edge $v_{q-1}v_{q}$ (i.e.~the
last edge of $P^{-}$) is colored black. Thus, since both $C_{0}$-edges $%
v_{q-1}v_{q}$ and $v_{\ell }v_{q}$ are black, the operation Black-Force is
triggered at vertex $v_{q}$ which colors the chord $v_{q}v_{q}^{\ast }$
yellow. Then, the operation Yellow-Force (which is triggered once $%
v_{q}v_{q}^{\ast }$ is colored yellow) colors at least one edge incident to $%
v_{q}^{\ast }$ black. Thus, so far we have at the current iteration at least
6 new forced edges, namely at least one forced edge in each of $P^{+}$ and $%
P^{-}$, as well as the edges $v_{\ell }v_{\ell }^{\ast },v_{\ell
}v_{q},v_{q}v_{q}^{\ast }$ and at least one $C_{0}$-edge incident to $%
v_{q}^{\ast }$.

If $v_{q}^{\ast }$ is incident to two uncolored $C_{0}$-edges then $%
v_{q}^{\ast }$ becomes a new biased vertex, i.e.~we have four new biased
vertices $v_{y},v_{q},v_{\ell },v_{q}^{\ast }$. If $v_{q}^{\ast }$ is
incident to two previously colored (i.e.~black) $C_{0}$-edges then the edges 
$v_{q}v_{q}^{\ast }$ and $v_{\ell }v_{q}$ have been colored at a previous
iteration yellow and black, respectively, which is a contradiction as they
received these colors at the current iteration. Finally, if $v_{q}^{\ast }$
is incident to one previously colored (i.e.~black) $C_{0}$-edge and to one
uncolored $C_{0}$-edge, then a new forcing path starts at $v_{q}$ with the
chord $v_{q}v_{q}^{\ast }$. At the end of this forcing path, there must be
at least one $C_{0}$-edge incident to an unbiased vertex $v_{z}$, which
becomes black after all Black-Force and Yellow-Force operations. Thus $v_{z}$
becomes a new biased vertex, that is, we have four new biased vertices $%
v_{y},v_{q},v_{\ell },v_{z}$.

Summarizing, in Case 1(ii) we have two instances $I_{1},I_{2}$, each having
at least 4 new biased vertices and 6 forced edges.

\emph{Case 1(iii).} $v_{\ell +1}=v_{q+1}$. Similarly to Case 1(ii), when
edge $v_{\ell -1}v_{\ell }$ is colored red, Red-Force is triggered which
colors the $C_{0}$-edge $v_{\ell }v_{\ell +1}$ black and the chord $v_{\ell
}v_{\ell }^{\ast }$ blue. However, since $v_{q}$ is the last vertex of $%
P^{-} $, it follows that the edge $v_{q}v_{q+1}=v_{q}v_{\ell +1}$ is
uncolored, and thus $v_{\ell +1}$ becomes a new biased vertex. That is, we
have 4 new biased vertices $v_{y},v_{q},v_{\ell },v_{\ell +1}$. Assume that $%
P^{+}\cup P^{-}$ has two edges. Then each of $P^{+}$ and $P^{-}$ has only
one edge, namely the $C_{0}$-edges $v_{y}v_{y+1}$ and $v_{y}v_{y-1}$,
respectively. Then, since the edges $v_{\ell }v_{\ell +1}$ and $v_{q}v_{\ell
+1}$ are also $C_{0}$-edges, it follows that we there is a cycle of 4 $C_{0}$%
-edges. This is a contradiction, as $C_{0}$ is a Hamiltonian cycle of a
graph $G$ with more than 4 vertices. Thus $P^{+}\cup P^{-}$ has at least 4
edges, and thus we have at least 6 new forced edges, namely at least 4
forced edges in $P^{+}\cup P^{-}$, as well as the edges $v_{\ell }v_{\ell
}^{\ast },v_{\ell }v_{\ell +1}$.

Summarizing, in Case 1(iii) we have two instances $I_{1},I_{2}$, each having
at least 4 new biased vertices and 6 forced edges.

\emph{Case 1(iv).} $\{v_{\ell },v_{\ell +1}\}\cap
\{v_{q},v_{q+1}\}=\emptyset $. Similarly to Cases 1(ii) and 1(iii), the
chord $v_{\ell }v_{\ell }^{\ast }$ is colored blue and the $C_{0}$-edge $%
v_{\ell }v_{\ell +1}$ is colored black by the Red-Force operation triggered
at the end of $P^{+}$. If $v_{\ell +1}v_{\ell +2}$ is uncolored then vertex $%
v_{\ell +1}$ is a new biased vertex, that is, we have 4 new biased vertices $%
v_{y},v_{q},v_{\ell },v_{\ell +1}$. Now assume that $v_{\ell +1}v_{\ell +2}$
is a previously colored (i.e.~black) $C_{0}$-edge. Then, at the end of $%
P^{+} $ we have a new forcing path starting at vertex $v_{\ell }$ with the
edge $v_{\ell }v_{\ell +1}$. At the end of this forcing path, there must be
at least one $C_{0}$-edge incident to an unbiased vertex $v_{z}$, which
becomes black after all Black-Force and Yellow-Force operations. Note that
vertex $v_{q}$ can possibly be one of the internal vertices of this new
forcing path. Thus $v_{z}$ becomes a new biased vertex, that is, we have 4
new biased vertices $v_{y},v_{q},v_{\ell },v_{z}$.

Suppose that \emph{both} edges $v_{\ell +1}v_{\ell +2}$ and $v_{q+1}v_{q+2}$
are previously \emph{not} colored black, i.e.~they are both uncolored. Then,
in both instances $I_{1},I_{2}$ we have at least 4 new forced edges, namely
at least one forced edge in each of $P^{+}$ and $P^{-}$, as well as the
edges $v_{\ell }v_{\ell }^{\ast },v_{\ell }v_{\ell +1}$ for $I_{1}$
(resp.~the edges $v_{q}v_{q}^{\ast },v_{q}v_{q+1}$ for~$I_{2}$). Note here
that this case is only possible to appear when, until the current iteration,
vertex $v_{y}$ is in the center of a path of at least \emph{6 consecutive $%
C_{0}$-edges} that have not been colored yet, namely the edges $%
v_{y-3}v_{y-2},v_{y-2}v_{y-1},v_{y-1}v_{y},v_{y}v_{y+1},v_{y+1}v_{y+2},v_{y+2}v_{y+3} 
$.

Now suppose that at least one of the edges $v_{\ell +1}v_{\ell +2}$ and $%
v_{q+1}v_{q+2}$ is a previously colored black edge, say this edge is $%
v_{\ell +1}v_{\ell +2}$ without loss of generality. Then, although in the
instance $I_{2}$ we have again at least 4 new forced edges, in the instance $%
I_{1}$ we additionally have the chord $v_{\ell +1}v_{\ell +1}^{\ast }$ and
at least one $C_{0}$-edge incident to $v_{\ell +1}^{\ast }$ which are
colored yellow and black, respectively. That is, $I_{1}$ has in this case at
least 6 new forced edges.

Summarizing, in Case~1(iv) we have two instances $I_{1},I_{2}$, either each
having at least 4 new biased vertices and 4 forced edges, or one instance
having at least 4 new biased vertices and 4 forces edges and the other
instance having at least 4 new biased vertices and 6 forced edges.

\medskip

\emph{Case 2.} Both $P^{-}$ and $P^{+}$ end with a chord. That is, each of
these forcing paths has an odd number of \emph{internal} vertices. Since the
analysis for both new instances $I_{1},I_{2}$ is symmetric, here we only
analyze instance~$I_{1}$, i.e.~the case where $v_{y}v_{y+1}$ becomes red and 
$v_{y}v_{y-1}$ becomes black by the operation Blue-Branch. Let $v_{\ell }$
be the last vertex of $P^{+}$ and let the chord $v_{\ell }^{\ast }v_{\ell }$
be the last edge of $P^{+}$. Similarly, let $v_{q}$ be the last vertex of $%
P^{-}$ and let the chord $v_{q}^{\ast }v_{q}$ be the last edge of $P^{-}$.
That is, $v_{\ell }^{\ast }$ and $v_{q}^{\ast }$ are the last \emph{internal}
vertices of $P^{+}$ and of $P^{-}$, respectively. Since $P^{+}$ and $P^{-}$
share no common internal vertices by Observation~\ref%
{forced-paths-disjoint-internal-vertices-obs}, it follows that $v_{\ell
}^{\ast }\neq v_{q}^{\ast }$, and thus also $v_{\ell }\neq v_{q}$ (as every
vertex is incident to exactly one chord). Note that $v_{q}$ becomes a new
biased vertex after the forcing operations along $P^{-}$, as it becomes
incident to the new black edges $v_{q}v_{q-1}$ and $v_{q}v_{q+1}$. That is, $%
v_{y}$ and $v_{q}$ become new biased vertices.

\emph{Case 2(i).} $v_{\ell }\in \{v_{q-1},v_{q+1}\}$, i.e.~$v_{\ell }v_{q}$
is one of the two $C_{0}$-edges that are incident to $v_{q}$. Assume without
loss of generality that $v_{\ell }=v_{q-1}$, or equivalently $v_{q}=v_{\ell
+1}$. Then the $C_{0}$-edges incident to $v_{\ell }$ are $v_{\ell }v_{q}$
and $v_{\ell }v_{\ell -1}$. In this case the operation Yellow-Force which is
triggered at the end of $P^{-}$ (once the chord~$v_{q}^{\ast }v_{q}$ is
colored yellow) colors both $C_{0}$-edges $v_{\ell }v_{q}$ and $v_{q}v_{q+1}$
black. Thus, since $v_{\ell }v_{q}$ becomes black and $v_{\ell }v_{\ell
}^{\ast }$ becomes blue, the operation Blue-Force colors the $C_{0}$-edge $%
v_{\ell }v_{\ell -1}$ red. Furthermore, the operation Red-Force colors the
chord $v_{\ell -1}v_{\ell -1}^{\ast }$ blue. That is, we have at least 8 new
forced edges, i.e.~at least two forced edges in each of $P^{+},P^{-}$, as
well as the edges $v_{\ell }v_{q},v_{q}v_{q+1},v_{\ell }v_{\ell -1},v_{\ell
-1}v_{\ell -1}^{\ast }$.

Furthermore, $v_{\ell }$ becomes a new biased vertex as it becomes incident
to a new black edge and a new red edge. If $v_{q+1}$ is not incident to any
previously colored black $C_{0}$-edge, then $v_{q+1}$ becomes a new biased
vertex, that is, we have at least 4 new biased vertices $v_{y},v_{q},v_{\ell
},v_{q+1}$. Otherwise, if $v_{q+1}$ is incident to a previously colored
black $C_{0}$-edge, then a new forcing path starts at $v_{q+1}$ with the
edge $v_{q+1}v_{q+1}^{\ast }$. At the end of this forcing path, there must
be at least one $C_{0}$-edge incident to an unbiased vertex $v_{z}$, which
becomes black after all Black-Force and Yellow-Force operations. Thus $v_{z}$
becomes a new biased vertex, that is, we have at least 4 new biased vertices 
$v_{y},v_{q},v_{\ell },v_{z}$.

Summarizing, in Case 2(i) we have two instances $I_{1},I_{2}$, each having
at least 4 new biased vertices and 8 forced edges.

\emph{Case 2(ii).} $v_{q}\in \{v_{\ell -1},v_{\ell +1}\}$. This case is
equivalent to Case 2(i).

\emph{Case 2(iii).} $\{v_{\ell },v_{\ell -1},v_{\ell +1}\}\cap
\{v_{q},v_{q-1},v_{q+1}\}=\emptyset $. Assume that both $v_{q-1},v_{q+1}$
are not incident to any previously colored black vertices. Then both $%
v_{q-1},v_{q+1}$ become new biased vertices, i.e.~we have in total 4 new
biased vertices $v_{y},v_{q},v_{q-1},v_{q+1}$. Furthermore we have at least
6 new forced edges, namely at least two forced edges in each of $P^{+}$ and $%
P^{-}$, as well as the edges $v_{q}v_{q-1},v_{q}v_{q+1}$.

Now assume that one of the vertices $v_{q-1},v_{q+1}$ (say $v_{q-1}$) is
incident to a previously colored black vertex and the other one (sat $%
v_{q+1} $) is not. Then, on the one hand, $v_{q+1}$ becomes a new biased
vertex. On the other hand, a new forcing path $P^{\#}$ starts at vertex $%
v_{q}$ with the edge $v_{q}v_{q-1}$. At the end of this forcing path $P^{\#}$%
, there must be at least one $C_{0}$-edge incident to an unbiased vertex $%
v_{z}$, which becomes black after all Black-Force and Yellow-Force
operations. Note that vertex $v_{q+1}$ can possibly be one of the internal
vertices of this new forcing path $P^{\#}$. Thus $v_{z}$ becomes a new
biased vertex, that is, we have 4 new biased vertices $%
v_{y},v_{q},v_{q+1},v_{z}$. Furthermore we have in this case at least 8 new
forced edges, namely at least two forced edges in each of $P^{+}$ and $P^{-}$%
, the edges $v_{q}v_{q-1},v_{q}v_{q+1}$, and at least two more new forced
edges (at least one chord and one $C_{0}$-edge) in the new forcing path $%
P^{\#}$.

Next assume that each of the vertices $v_{q-1},v_{q+1}$ is incident to a
previously colored black vertex. Then, at vertex $v_{q}$ we have either two
new forcing paths (one starting with the edge $v_{q}v_{q-1}$ and one
starting with the edge $v_{q}v_{q+1}$) or one forcing cycle (containing the
edges $v_{q}v_{q-1}$ and $v_{q}v_{q+1}$). If we have one forcing cycle, it
must contain at least 5 new forced edges (as $G$ is without loss of
generality triangle-free by Theorem~\ref{triangle-free-thm}). This implies
that we have in total at least two 2 new biased vertices $v_{y},v_{q}$ and
at least 9 new forced edges, namely at least two forced edges in each of $%
P^{+}$ and $P^{-}$, as well as at least 5 more forced edges in the forcing
cycle at $v_{q}$.

Finally assume that we have two forcing paths $P^{\#}$ and $P^{\ast }$,
starting with the edge $v_{q}v_{q-1}$ and with the edge $v_{q}v_{q+1}$,
respectively. At the end of each of these forcing paths $P^{\#}$ and $%
P^{\ast }$, there must be at least one $C_{0}$-edge incident to an unbiased
vertex $v_{z^{\#}}$ and $v_{z^{\ast }}$, respectively, which becomes black
after all Black-Force and Yellow-Force operations. Thus each of the vertices 
$v_{z^{\#}}$ and $v_{z^{\ast }}$ becomes a new biased vertex. Assume that $%
v_{z^{\#}} \neq v_{z^{\ast }}$. Then we have 4 new biased vertices $%
v_{y},v_{q},v_{z^{\#}},v_{z^{\ast }}$. Furthermore we have at least 10 new
forced edges, namely at least two forced edges in each of $P^{+}$ and $P^{-}$%
, the edges $v_{q}v_{q-1},v_{q}v_{q+1}$, and at least two more new forced
edges (at least one $C_{0}$-edge and one chord) in each of the new forcing
paths $P^{\#}$ and~$P^{\ast }$. Now assume that $v_{z^{\#}} = v_{z^{\ast}}$.
Then there exists a new forcing path $P^{**}$ starting at $v_{z^{\#}}$. At
the end of this forcing path $P^{**}$ there must be at least one $C_{0}$%
-edge incident to an unbiased vertex $v_{z}$, which becomes black after all
Black-Force and Yellow-Force operations. Thus $v_{z}$ becomes a new biased
vertex, that is, we have 4 new biased vertices $v_{y},v_{q},v_{z^{\#}},v_{z}$%
. Furthermore, we have in this case at least 10 new forced edges, namely at
least two forced edges in each of $P^{+}$ and $P^{-}$, at least 5 forced
edges in $P^{\#} \cup P^{\ast }$, and at least one more forced edge in $%
P^{**}$.

Summarizing, in Case 2(iii) we have two instances $I_{1},I_{2}$, each of
them having either at least 2 new biased vertices and 9 forced edges, or at
least 4 new biased vertices and 6 forced edges.

\medskip

\emph{Case 3.} $P^{-}$ ends with a chord and $P^{+}$ ends with a $C_{0}$%
-edge. That is, $P^{-}$ has an odd number and $P^{+}$ has an even number of 
\emph{internal} vertices. Here the analysis for the two new instances $%
I_{1},I_{2}$ is not symmetric, so we will analyze them separately.

\emph{Case 3(i).} $P^{+}\cup P^{-}$ builds a $P_{4}$ (i.e.~a path with 4
vertices) whose two endpoints are adjacent. Then the algorithm executes line~%
\ref{alg-line-21} and calls Procedure~\ref{P4-proc}. As $P^{+}\cup P^{-}$
builds a $P_{4}$, either $P^{+}$ contains one edge and $P^{-}$ contains two
edges (see lines~\ref{proc-2-line-2}-\ref{proc-2-line-4} of Procedure~\ref%
{P4-proc}), or $P^{+}$ contains two edges and $P^{-}$ contains one edge (see
lines~\ref{proc-2-line-5}-\ref{proc-2-line-7} of Procedure~\ref{P4-proc}).
As the analysis of both these cases is symmetric, we only consider here the
first case, i.e.~that $P^{+}$ contains one edge and $P^{-}$ contains two
edges. In this case the algorithm branches to two new instances $I_{1},I_{2}$%
, as follows. In $I_{1}$, the $C_{0}$-edge $v_{y}v_{y+1}$ becomes red, the $%
C_{0}$-edges $v_{y}v_{y-1},v_{y+1}v_{y+2},v_{y+2}v_{y+3}$ become black, the
chord $v_{y+1}v_{y+1}^{\ast }$ becomes blue, and the chord $v_{y-1}v_{y+2}$
becomes yellow. Furthermore the 4 edges $%
v_{y}v_{y+1},v_{y}v_{y-1},v_{y-1}v_{y+2},v_{y+2}v_{y+1}$ are added as an
ambivalent quadruple in the set $Q$ (see line~\ref{proc-2-line-3} of
Procedure~\ref{P4-proc}). In $I_{2}$, the $C_{0}$-edges $%
v_{y}v_{y-1},v_{y+2}v_{y+3}$ become red, the $C_{0}$-edges $%
v_{y}v_{y+1},v_{y+1}v_{y+2}$ become black, the chords $%
v_{y-1}v_{y+2},v_{y+3}v_{y+3}^{\ast }$ become blue, and the chord $%
v_{y+1}v_{y+1}^{\ast }$ becomes yellow. Furthermore, at least one more
(previously uncolored) incident $C_{0}$-edge of $v_{y+1}^{\ast }$ becomes
black by the Yellow-Force operation that is triggered once $%
v_{y+1}v_{y+1}^{\ast }$ becomes yellow. Thus, $I_{1}$ has at least 6 new
forced edges and $I_{2}$ has at least 8 new forced edges.

In $I_{1}$, if $v_{y+3}$ is not incident to any previously colored black
edge, it becomes a new biased vertex. That is, in this case we have in $%
I_{1} $ in total at least 4 new biased vertices $%
v_{y},v_{y+1},v_{y+2},v_{y+3}$. Otherwise, if $v_{y+3}$ is incident to a
previously colored black edge, a new forcing path starts at $v_{y+3}$. At
the end of this forcing path there must be at least one $C_{0}$-edge
incident to an unbiased vertex $v_{z}$, which becomes black after all
Black-Force and Yellow-Force operations. Thus $v_{z}$ becomes a new biased
vertex, that is, we have 4 new biased vertices $v_{y},v_{y+1},v_{y+2},v_{z}$.

In $I_{2}$, if $v_{y+1}^{\ast }$ is not incident to any previously colored
black edge, it becomes a new biased vertex. That is, in this case we have in 
$I_{2}$ in total at least 4 new biased vertices $%
v_{y},v_{y+1},v_{y+2},v_{y+1}^{\ast }$. Otherwise, if $v_{y+1}^{\ast }$ is
incident to a previously colored black edge, a new forcing path starts at $%
v_{y+1}^{\ast }$. At the end of this forcing path there must be at least one 
$C_{0}$-edge incident to an unbiased vertex $v_{z}$, which becomes black
after all Black-Force and Yellow-Force operations. Thus $v_{z}$ becomes a
new biased vertex, that is, we have 4 new biased vertices $%
v_{y},v_{y+1},v_{y+2},v_{z}$.

Summarizing, in Case 3(i), $I_{1}$ has at least 4 new biased vertices and 6
new forced edges, while $I_{2}$ has at least 4 new biased vertices and 8 new
forced edges.

\emph{Case 3(ii).} $P^{+}\cup P^{-}$ does not build a $P_{4}$ whose two
endpoints are adjacent. Since $P^{-}$ ends with a chord and $P^{+}$ ends
with a $C_{0}$-edge by the assumption of Case 3, it easily follows that $%
P^{+}\cup P^{-}$ does also not byuld a $C_{4}$. Therefore Algorithm~\ref%
{alt-cycle-detection-alg} executes lines~\ref{alg-line-22}-\ref{alg-line-26} and
branches to two new instances $I_{1},I_{2}$, which we analyze separately
below. Let $v_{\ell }$ be the last vertex of $P^{+}$, and assume without
loss of generality that the last edge of $P^{+}$ is $v_{\ell -1}v_{\ell }$
(the other case, where the last edge of $P^{+}$ is $v_{\ell }v_{\ell +1}$ is
exactly symmetric). Furthermore, let $v_{q}$ be the last vertex of $P^{-}$
and let the chord $v_{q}^{\ast }v_{q}$ be the last edge of $P^{-}$. That is, 
$v_{\ell -1}$ and $v_{q}^{\ast }$ are the last \emph{internal} vertices of $%
P^{+}$ and of $P^{-}$, respectively.

\emph{Case 3(ii)(a).} $v_{\ell }\in \{v_{q-1},v_{q+1}\}$, i.e.~$v_{\ell
}v_{q}$ is one of the two $C_{0}$-edges that are incident to $v_{q}$. Assume
without loss of generality that $v_{\ell }=v_{q-1}$, or equivalently $%
v_{q}=v_{\ell +1}$. Then the $C_{0}$-edges incident to $v_{\ell }$ are $%
v_{\ell }v_{q}$ and $v_{\ell }v_{\ell -1}$. Note that, if $P^{+}\cup P^{-}$
has three edges, then $P^{+}\cup P^{-}$ builds a $P_{4}$ whose two endpoints
are adjacent, which is a contradiction to the assumption of Case 3(ii). Thus 
$P^{+}\cup P^{-}$ has at least 5 edges. {}

In the instance $I_{1}$, the operation Yellow-Force which is triggered at
the end of $P^{-}$ (once the chord $v_{q}^{\ast }v_{q}$ is colored yellow)
colors both $C_{0}$-edges $v_{\ell }v_{q}$ and $v_{q}v_{q+1}$ black.
Furthermore, when the $C_{0}$-edge $v_{\ell -1}v_{\ell }$ becomes red, the
operation Red-Force colors the chord $v_{\ell }v_{\ell }^{\ast }$ blue. That
is, we have at least 8 new forced edges, namely at least 5 forced edges in $%
P^{+}\cup P^{-}$, as well as the edges $v_{\ell }v_{\ell }^{\ast },v_{\ell
}v_{q},v_{q}v_{q+1}$. If $v_{q+1}v_{q+2}$ is not a black edge, then $v_{q+1}$
becomes a new biased vertex, that is, we have the 4 biased vertices $%
v_{y},v_{\ell },v_{q},v_{q+1}$. Otherwise, if $v_{q+1}v_{q+2}$ is a black
edge, then a new forcing path starts at $v_{q+1}$ with the chord $%
v_{q+1}v_{q+1}^{\ast }$. At the end of this forcing path, there must be at
least one $C_{0}$-edge incident to an unbiased vertex $v_{z}$, which becomes
black after all Black-Force and Yellow-Force operations. Thus $v_{z}$
becomes a new biased vertex, that is, we have 4 new biased vertices $%
v_{y},v_{\ell },v_{q},v_{z}$.

In the instance $I_{2}$, recall that $v_{y}v_{y+1}$ becomes black and $%
v_{y}v_{y-1}$ becomes red by the operation Blue-Branch. Note that here $%
v_{\ell }$ becomes a new biased vertex after the forcing operations along $%
P^{+}$, as it becomes incident to the new black edge $v_{\ell -1}v_{\ell }$.
That is, we have at least the 2 new biased vertices $v_{y},v_{\ell }$.
Furthermore we have at least 5 new forced edges, namely the edges $P^{+}\cup
P^{-}$.

Summarizing, in Case 3(ii)(a), $I_{1}$ has at least 4 new biased vertices
and 8 forced edges, while $I_{2}$ has at least 2 new biased vertices and 5
forced edges.

\emph{Case 3(ii)(b).} $v_{\ell }\notin \{v_{q-1},v_{q+1}\}$. In the instance 
$I_{2}$ (i.e.~where the $C_{0}$-edges $v_{y}v_{y+1}$ and $v_{y}v_{y-1}$
become black and red, respectively) $v_{\ell }$ becomes a new biased vertex
after the forcing operations along $P^{+}$, as it becomes incident to the
new black edge $v_{\ell -1}v_{\ell }$. That is, we have at least the 2 new
biased vertices $v_{y},v_{\ell }$. Furthermore we have at least 3 new forced
edges, namely at least one forced edge in $P^{+}$ and at least two forced
edges in $P^{-}$. That is, in Case 3(ii)(b), $I_{2}$ has at least 2 new
biased vertices and 3 forced edges.

\emph{Analysis of instance }$I_{1}$\emph{.} Recall that here $v_{y}v_{y+1}$
becomes red and that $v_{y}v_{y-1}$ becomes black by the operation
Blue-Branch. Note that $v_{q}$ becomes a new biased vertex after the forcing
operations along $P^{-}$, as it becomes incident to the new black edges $%
v_{q}v_{q-1}$ and $v_{q}v_{q+1}$. Moreover, note that $v_{\ell }$ becomes a
new biased vertex after the forcing operations along $P^{+}$, as it becomes
incident to the new red edge $v_{\ell -1}v_{\ell }$ and to the new black
edge $v_{\ell }v_{\ell +1}$. That is, $v_{y},v_{\ell },v_{q}$ become new
biased vertices.

\emph{Case 3(ii)(b)(1).} $v_{\ell }=v_{q}$. In this case, after the forcing
operations along $P^{-}$, vertex $v_{\ell }$ is incident to the two black $%
C_{0}$-edges $v_{q}v_{q-1},v_{q}v_{q+1}$, and thus the edge $v_{\ell
-1}v_{\ell }$ cannot be colored red in the forcing path $P^{+}$, which is a
contradiction.

\emph{Case 3(ii)(b)(2).} $v_{\ell +1}\in \{v_{q-1},v_{q+1}\}$. Then, since
the case $v_{\ell +1}=v_{q+1}$ is equivalent to the case $v_{\ell }=v_{q}$
(which has been dealt with in Case 3($I_{1}$)(i)), we assume that $v_{\ell
+1}=v_{q-1}$. The Red-Force operation, which is triggered when the last edge 
$v_{\ell -1}v_{\ell }$ of $P^{+}$ becomes red, colors the $C_{0}$-edge $%
v_{\ell }v_{\ell +1}$ black and the chord $v_{\ell }v_{\ell }^{\ast }$ blue.
Furthermore the Yellow-Force operation, which is triggered when the last
edge $v_{q}^{\ast }v_{q}$ of $P^{-}$ becomes yellow, colors both $C_{0}$%
-edges $v_{q}v_{q-1}=v_{q}v_{\ell +1}$ and $v_{q}v_{q+1}$ black. Thus vertex 
$v_{\ell +1}$ becomes a new biased vertex. That is, we have at least 4 new
biased vertices $v_{y},v_{\ell },v_{q},v_{\ell +1}$.

Assume that $v_{q+1}$ is incident to a previously colored black edge. Then,
once $v_{q}v_{q+1}$ is colored black, the operation Black-Force is triggered
which colors the chord $v_{q+1}v_{q+1}^{\ast }$ yellow. On the other hand,
once $v_{\ell }v_{\ell +1}$ and $v_{q}v_{\ell +1}$ are colored black, the
Black-Force and Yellow-Force operations are triggered, and thus the chord $%
v_{\ell +1}v_{\ell +1}^{\ast }$ becomes yellow and at least one of the two $%
C_{0}$-edges insicent to $v_{\ell +1}^{\ast }$ becomes black. Thus we have
at least 10 new forced edges, namely at least one forced edge in $P^{+}$, at
least two forced edges in $P^{-}$, the edges $v_{\ell }v_{\ell }^{\ast
},v_{\ell }v_{\ell +1},v_{q}v_{\ell +1},v_{q}v_{q+1},v_{q+1}v_{q+1}^{\ast
},v_{\ell +1}v_{\ell +1}^{\ast }$, as well as at least one $C_{0}$-edge
incident to $v_{\ell +1}^{\ast }$.

Now assume that $v_{q+1}$ is not incident to any previously colored black
edge. Then, once $v_{\ell }v_{\ell +1}$ and $v_{q}v_{\ell +1}$ are colored
black, the Black-Force operation is triggered, which colors the chord $%
v_{\ell +1}v_{\ell +1}^{\ast }$ yellow. Let $v_{t}=v_{\ell +1}^{\ast }$,
i.e.~the two $C_{0}$-edges incident to $v_{\ell +1}^{\ast }$ are $v_{\ell
+1}^{\ast }v_{t-1}$ and $v_{\ell +1}^{\ast }v_{t+1}$. Note that at least one
of these edges $v_{\ell +1}^{\ast }v_{t-1}$ and $v_{\ell +1}^{\ast }v_{t+1}$
is uncolored, as otherwise $v_{\ell +1}v_{\ell +1}^{\ast }$ would have been
colored yellow at a previous iteration, which is a contradiction.
Furthermore note that both vertices $v_{t-1}$ and $v_{t+1}$ are different
than $v_{q}$, since otherwise $G$ would have a triangle on the vertices $%
v_{\ell +1},v_{\ell +1}^{\ast },v_{q}$, which is a contradiction by Theorem~\ref{triangle-free-thm}.

Suppose that $v_{q+1}\in \{v_{t-1},v_{t+1}\}$, say without loss of
generality that $v_{q+1}=v_{t-1}$. Then the edge $v_{\ell +1}^{\ast
}v_{t-1}=v_{\ell +1}^{\ast }v_{q+1}$ is currently uncolored, as we assumed
that $v_{q+1}$ is not incident to any previously colored black edge. Once $%
v_{\ell +1}v_{\ell +1}^{\ast }$ is colored yellow, the operation
Yellow-Force is triggered which colors the edge $v_{\ell +1}^{\ast }v_{q+1}$
black. Thus, as both $v_{\ell +1}^{\ast }v_{q+1}$ and $v_{q}v_{q+1}$ are
colored black, the Black-Force operation is triggered which colors the chord 
$v_{q+1}v_{q+1}^{\ast }$ yellow. Thus we have again at least 10 new forced
edges, namely at least one forced edge in $P^{+}$, at least two forced edges
in $P^{-}$, and the edges $v_{\ell }v_{\ell }^{\ast },v_{\ell }v_{\ell
+1},v_{q}v_{\ell +1},v_{q}v_{q+1},v_{q+1}v_{q+1}^{\ast },v_{\ell +1}v_{\ell
+1}^{\ast },v_{\ell +1}^{\ast }v_{q+1}$.

Finally suppose that $v_{q+1}\notin \{v_{t-1},v_{t+1}\}$. If both edges $%
v_{\ell +1}^{\ast }v_{t-1}$ and $v_{\ell +1}^{\ast }v_{t+1}$ are uncolored,
they both become black by the Yellow-Force operation that is triggered at $%
v_{\ell +1}^{\ast }$, once $v_{\ell +1}v_{\ell +1}^{\ast }$ becomes yellow.
Thus in this case we have at least 10 new forced edges, namely at least one
forced edge in $P^{+}$, at least two forced edges in $P^{-}$, and the edges $%
v_{\ell }v_{\ell }^{\ast },v_{\ell }v_{\ell +1},v_{q}v_{\ell
+1},v_{q}v_{q+1},v_{\ell +1}v_{\ell +1}^{\ast },v_{\ell +1}^{\ast
}v_{t-1},v_{\ell +1}^{\ast }v_{t+1}$. Otherwise, let one of the edges $%
v_{\ell +1}^{\ast }v_{t-1}$ and $v_{\ell +1}^{\ast }v_{t+1}$ (say, the edge $%
v_{\ell +1}^{\ast }v_{t-1}$) be previously colored black. Then we have at
least 9 new forced edges, namely at least one forced edge in $P^{+}$, at
least two forced edges in $P^{-}$, and the edges $v_{\ell }v_{\ell }^{\ast
},v_{\ell }v_{\ell +1},v_{q}v_{\ell +1},v_{q}v_{q+1},v_{\ell +1}v_{\ell
+1}^{\ast },v_{\ell +1}^{\ast }v_{t+1}$. If the $C_{0}$-edge $v_{t+1}v_{t+2}$
is uncolored, then $v_{t+1}$ is a new biased vertex, and thus we have at
least 5 new biased vertices $v_{y},v_{\ell },v_{q},v_{\ell +1},v_{t+1}$.
Otherwise, if $v_{t+1}v_{t+2}$ is black, then a new forcing path starts at $%
v_{\ell +1}^{\ast }$ with the edge $v_{\ell +1}^{\ast }v_{t+1}$. At the end
of this forcing path, there must be at least one $C_{0}$-edge incident to an
unbiased vertex $v_{z}$, which becomes black after all Black-Force and
Yellow-Force operations. Thus $v_{z}$ becomes a new biased vertex, that is,
we have again at least 5 new biased vertices $v_{y},v_{\ell },v_{q},v_{\ell
+1},v_{z}$.

Summarizing, in Case 3(ii)(b)(2), $I_{1}$ has either at least 4 new biased
vertices and 10 new forced edges, or at least 5 new biased vertices and 9
new forced edges.

\emph{Case 3(ii)(b)(3).} $v_{\ell +1}=v_{q}$. This case is equivalent to
Case 3(ii)(a).

\emph{Case 3(ii)(b)(4).} $\{v_{\ell },v_{\ell +1}\}\cap
\{v_{q},v_{q-1},v_{q+1}\}=\emptyset $. First suppose that none of the three
vertices $v_{\ell +1},v_{q-1},v_{q+1}$ is incident to any previously colored
black edge. Then all these three vertices become new biased vertices, and
thus we have in total at least 6 new biased vertices (i.e.~together with $%
v_{y},v_{\ell },v_{q}$). Furthermore, in this case we have at least 7 new
forced edges, namely at least one forced edge in $P^{+}$, at least two
forced edges in $P^{-}$, as well as the four edges $v_{\ell }v_{\ell }^{\ast
},v_{\ell }v_{\ell +1},v_{q}v_{q-1},v_{q}v_{q+1}$.

Suppose that exactly one of the three vertices $v_{\ell +1},v_{q-1},v_{q+1}$
is incident to a previously colored black edge; denote this vertex by $v_{t}$%
. Then, the two vertices in $\{v_{\ell +1},v_{q-1},v_{q+1}\}\setminus
\{v_{t}\}$ are new biased vertices, i.e.~we have in total at least 5 new
biased vertices (together with $v_{y},v_{\ell },v_{q}$). Furthermore, the
Black-Force operation is triggered at $v_{t}$, which colors the chord $%
v_{t}v_{t}^{\ast }$ yellow. Moreover the Yellow-Force operation is forced at 
$v_{t}^{\ast }$, which colors at least one of the $C_{0}$-edges incident to $%
v_{t}^{\ast }$ black. Thus we have at least 9 new forced edges, namely at
least one forced edge in $P^{+}$, at least two forced edges in $P^{-}$, as
well as the edges $v_{\ell }v_{\ell }^{\ast },v_{\ell }v_{\ell
+1},v_{q}v_{q-1},v_{q}v_{q+1},v_{t}v_{t}^{\ast }$ and at least one of the $%
C_{0}$-edges incident to $v_{t}^{\ast }$.

Now suppose that exactly two of the three vertices $v_{\ell
+1},v_{q-1},v_{q+1}$ are incident to a previously colored black edge; denote
these vertices by $v_{t_{1}},v_{t_{2}}$. Then, the single vertex in $%
\{v_{\ell +1},v_{q-1},v_{q+1}\}\setminus \{v_{t_{1}}v_{t_{2}}\}$ is a new
biased vertex, i.e.~we have in total at least 4 new biased vertices
(together with $v_{y},v_{\ell },v_{q}$). Furthermore, the Black-Force
operation is triggered both at $v_{t_{1}}$ and at $v_{t_{2}}$, which color
the chords $v_{t_{1}}v_{t_{1}}^{\ast }$ and $v_{t_{2}}v_{t_{2}}^{\ast }$
yellow. Moreover, the Yellow-Force operations that are forced at $%
v_{t_{1}}^{\ast }$ and $v_{t_{2}}^{\ast }$ colors at least one of the $C_{0}$%
-edges incident to $v_{t_{1}}^{\ast }$ or $v_{t_{2}}^{\ast }$ black (note
that this $C_{0}$-edge may be incident to both $v_{t_{1}}^{\ast }$ and $%
v_{t_{2}}^{\ast }$). Thus we have at least 10 new forced edges, namely at
least one forced edge in $P^{+}$, at least two forced edges in $P^{-}$, as
well as the edges $v_{\ell }v_{\ell }^{\ast },v_{\ell }v_{\ell
+1},v_{q}v_{q-1},v_{q}v_{q+1},v_{t_{1}}v_{t_{1}}^{\ast
},v_{t_{2}}v_{t_{2}}^{\ast }$ and at least one of the $C_{0}$-edges incident
to $v_{t_{1}}^{\ast }$ or $v_{t_{2}}^{\ast }$.

Finally suppose that all three vertices $v_{\ell +1},v_{q-1},v_{q+1}$ are
incident to a previously colored black edge. Then the Black-Force operations
that are triggered at these three vertices which color the chords $v_{\ell
+1}v_{\ell +1}^{\ast }$, $v_{q-1}v_{q-1}^{\ast }$, and $v_{q+1}v_{q+1}^{\ast
}$ yellow. Thus we have at least 10 new forced edges, namely at least one
forced edge in $P^{+}$, at least two forced edges in $P^{-}$, as well as the
edges $v_{\ell }v_{\ell }^{\ast },v_{\ell }v_{\ell
+1},v_{q}v_{q-1},v_{q}v_{q+1},v_{\ell +1}v_{\ell +1}^{\ast
},v_{q-1}v_{q-1}^{\ast },v_{q+1}v_{q+1}^{\ast }$. Furthermore, at the end of
at least one of the new forcing paths starting at the vertices $v_{\ell +1}$%
, $v_{q-1}$, and $v_{q+1}$ with the edges $v_{\ell +1}v_{\ell +1}^{\ast }$, $%
v_{q-1}v_{q-1}^{\ast }$, and $v_{q+1}v_{q+1}^{\ast }$, respectively, there
must be at least one $C_{0}$-edge incident to an unbiased vertex $v_{z}$,
which becomes black after all Black-Force and Yellow-Force operations. Thus $%
v_{z}$ becomes a new biased vertex, that is, we have again at least 4 new
biased vertices $v_{y},v_{\ell },v_{q},v_{z}$.

Summarizing, in Case 3(ii)(b)(4), $I_{1}$ has either at least 6 new biased
vertices and 7 new forced edges, or at least 5 new biased vertices and 9 new
forced edges, or at least 4 new biased vertices and 10 new forced edges.
\end{proof}

\medskip

We are now ready to use the results of our technical Lemma~\ref%
{alt-cycles-enum-lem} to derive an upper bound for the running time of
Algorithm~\ref{alt-cycle-detection-alg}.

\begin{theorem}
\label{exact-algorithm-time-thm}Let $G$ be a Smith graph on $n$ vertices
with a given Hamiltonian cycle $C_{0}$. 
Then Algorithm~\ref{alt-cycle-detection-alg} runs in $O(n\cdot 2^{0.299862744n}) = O(1.23103^n)$ time and in linear space. 
If~$G$ does not contain any induced cycle~$C_{6}$ on 6 vertices, then the running time becomes $O(n\cdot 2^{0.2971925n}) = O(1.22876^n)$. 
\end{theorem}

\begin{proof}
To derive the running time of the algorithm, we first upper-bound the number
of instances the algorithm produces using the inequalities in the statement
of Lemma~\ref{alt-cycles-enum-lem}. These inequalities upper bound the sizes
of the sets $W(I_{1}),W(I_{2})$ of unbiased vertices and of the sets $%
U(I_{1}),U(I_{2})$ of unforced edges in the two instances $I_{1},I_{2}$
which are obtained in the various cases where the algorithm branches. In
fact, for upper-bounding the instances produced by the algorithm, we can
ignore the cases where the algorithm just updates the current instance $I$
(instead of branching to two new instances $I_{1},I_{2}$), as in these cases
the current instance shrinks in a constant number of steps by Lemma~\ref%
{alt-cycles-enum-lem}. In our analysis below we assume that the
input graph has no $X$-certificate, as otherwise a second Hamiltonian cycle
is found in polynomial time by executing lines~\ref{alg-line-X-certif-1}-\ref{alg-line-X-certif-2} of
Algorithm~\ref{alt-cycle-detection-alg}.

Given an instance $I$ at some iteration of the algorithm, where $|W(I)|=x$
and $|U(I)|=y$, we denote by $f(x,y)$ the worst-case running time needed for
the algorithm to compute all the desired alternating red-blue \emph{cycles} $%
D=Red\cup Blue $ of $I$, (or to announce ``contradiction'' in the branches
that such a cycle $D$ does not exist). Recall that, with the help of the
ambivalent quadruples, these alternating cycles $D$ computed by the
algorithm form a succinct encoding of all possible alternating red-blue
cycles.

In any of the inequalities of Lemma~\ref{alt-cycles-enum-lem}, whenever the
instance $I$ is replaced by the instances $I_1,I_2$ and the values $x,y$ in $%
I$ are replaced by the values $x-k_1,y-\ell_1$ and $x-k_2,y-\ell_2$ in $I_1$
and $I_2$, respectively, then the running time $f(x,y)$ of the algorithm at $%
I$ is replaced by the sum $f(x-k_1,y-\ell_1) + f(x-k_2,y-\ell_2)$ of the
running times at $I_1$ and $I_2$, respectively. That is, in the worst case, $%
f(x,y) \leq f(x-k_1,y-\ell_1) + f(x-k_2,y-\ell_2)$. However, in order to
compute an \emph{upper bound} for $f(x,y)$ from all the recurrences, we need
to substitute ``$\leq$'' by ``$\geq$'' in each inequality. Thus we obtain
the following system of recurrence inequalities immediately by the statement
of Lemma~\ref{alt-cycles-enum-lem}.

\begin{enumerate}
\item $f(x,y)\geq 2f(x-2,y-7)$,

\item $f(x,y)\geq 2f(x-4,y-4)$,

\item $f(x,y)\geq 2f(x-2,y-9)$,

\item $f(x,y)\geq f(x-4,y-4)+f(x-4,y-6)$,

\item $f(x,y)\geq f(x-2,y-9)+f(x-4,y-6)$,

\item $f(x,y)\geq f(x-2,y-5)+f(x-4,y-8)$,

\item $f(x,y)\geq f(x-2,y-3)+f(x-6,y-7)$,

\item $f(x,y)\geq f(x-2,y-3)+f(x-4,y-10)$,

\item $f(x,y)\geq f(x-2,y-3)+f(x-5,y-9)$.
\end{enumerate}

For the sake of presentation, we divide the analysis of the above
recurrences into three parts. In Part A we give a naive upper bound on the
solution of the above recurrences in worst case. In Part B we give an upper
bound in the case where the input Smith graph $G$ does not contain any
induced cycle $C_6$ on 6 vertices. In Part C we present a more sophisticated
analysis of the worst case (including $C_6$'s), thus obtaining an improved
upper bound.

\medskip

\emph{Part A: A naive upper bound of the worst case.} To solve this system
of recurrences, we set $f(x,y)=2^{\alpha x+\beta y}$ and we compute the
optimum values for $\alpha ,\beta $ that satisfy all the above inequalities.
As it can be easily verified, the values $\{\alpha =0.15,\ \beta =0.1\}$
satisfy all these inequalities together. Moreover it can be easily checked
that, with these values of $\alpha ,\beta $, the first two relations become 
\emph{equalities}, while all the other relations become \emph{strict
inequalities}. Since $x\leq n$ and $y\leq \frac{3}{2}n$, it follows that $%
\alpha x+\beta y\leq (\alpha +\frac{3}{2}\beta )n=0.3\cdot n$, and thus the
algorithm produces at the end at most $2^{0.3\cdot n}$ different instances.
This completes the analysis of Part A.

\medskip

\emph{Part B: The case of $C_{6}$-free graphs.} The first inequality $%
f(x,y)\geq 2f(x-2,y-7)$ can only occur when the algorithm branches according
to Case~1(i) in the proof of Lemma~\ref{alt-cycles-enum-lem}, and in
particular when $P^{+}\cup P^{-}$ is a cycle with \emph{exactly} 6 edges.
Now assume that the input Smith graph $G$ does not contain any induced cycle
of 6 vertices. Then the first of the above inequalities never occurs, and
thus, in order to upper-bound the running time of the algorithm on $C_{6}$%
-free graphs, it suffices to solve the above recurrence system only for the
inequalities~$2,3,\ldots ,9$. To this end, we set again $f(x,y)=2^{\alpha
x+\beta y}$. As it can be easily verified, the values $\{\alpha =0.155615,\
\beta =0.094385\}$ satisfy all these inequalities together. Therefore, as $%
\alpha x+\beta y\leq (\alpha +\frac{3}{2}\beta )n=0.2971925\cdot n$, and
thus the total number of instances produced by the algorithm on $C_{6}$-free
graphs is upper-bounded by $2^{0.2971925\cdot n} = O(1.22876^n)$. This completes the
analysis of Part B.

\medskip

\emph{Part C: A more careful analysis of the worst case.} Here we introduce
two new parameters $k,\ell $ to the inequalities 1-9, where $k$ (resp.~$\ell 
$) denotes the number of subsequent potential applications of the first
(resp.~second) inequality. Then, the inequalities 1-9 become as follows.

\begin{enumerate}
\item[$1^{*}$.] $f(x,y,k,\ell )\geq 2f(x-2,y-7,k-1,\ell )$,

\item[$2^{*}$.] $f(x,y,k,\ell )\geq 2f(x-4,y-4,k,\ell -1)$,

\item[$3^{*}$.] $f(x,y,k,\ell )\geq 2f(x-2,y-9,k,\ell )$,

\item[$4^{*}$.] $f(x,y,k,\ell )\geq f(x-4,y-4,k,\ell )+f(x-4,y-6,k,\ell )$,

\item[$5^{*}$.] $f(x,y,k,\ell )\geq f(x-2,y-9,k,\ell )+f(x-4,y-6,k,\ell )$,

\item[$6^{*}$.] $f(x,y,k,\ell )\geq f(x-2,y-5,k,\ell )+f(x-4,y-8,k,\ell )$,

\item[$7^{*}$.] $f(x,y,k,\ell )\geq f(x-2,y-3,k,\ell )+f(x-6,y-7,k,\ell )$,

\item[$8^{*}$.] $f(x,y,k,\ell )\geq f(x-2,y-3,k,\ell )+f(x-4,y-10,k,\ell )$,

\item[$9^{*}$.] $f(x,y,k,\ell )\geq f(x-2,y-3,k,\ell )+f(x-5,y-9,k,\ell )$.
\end{enumerate}

The second inequality $f(x,y,k,\ell )\geq 2f(x-4,y-4,k,\ell -1)$ can only
occur when the algorithm branches according to Case~1(iv) in the proof of
Lemma~\ref{alt-cycles-enum-lem}. Recall by the proof of Lemma~\ref%
{alt-cycles-enum-lem} that this case is only possible to appear at a vertex $%
v_{y}$ when, until the current iteration, $v_{y}$ is in the center of a path
of at least \emph{6 consecutive $C_{0}$-edges} that have not been colored
yet. Then, in each of the two resulting new instances $I_{1},I_{2}$, exactly
one new chord and three new consecutive $C_{0}$-edges are forced (which are
colored black, red, and black, in this order), while $v_{y}$ is an internal
vertex of this path of the three $C_{0}$-edges. For simplicity of the
presentation, let us call these three $C_{0}$-edges the \emph{imperative path%
} of this application of the second recursion. Note that, in any of the
produced instances, between any two imperative paths there is at least one
edge that does not belong to any imperative path.

Recall by the analysis of Part B that the first inequality $f(x,y,k,\ell
)\geq 2f(x-2,y-7,k-1,\ell )$ can only occur when the algorithm branches
according to Case~1(i) in the proof of Lemma~\ref{alt-cycles-enum-lem},
where $P^{+}\cup P^{-}$ is a cycle~$C_{6}$ with \emph{exactly} 6 edges. In
this case, exactly four $C_{0}$-edges and three chords are being forced
(i.e.~colored), while these four newly forced $C_{0}$-edges are either two
pairs of consecutive $C_{0}$-edges, or three consecutive $C_{0}$-edges and
one separate $C_{0}$-edge (i.e.~not consecutive with the other three ones).
Furthermore, recall from the proof of Lemma~\ref{alt-cycles-enum-lem} that,
in this case, each of the four internal vertices of $P^{+}$ and $P^{-}$ is
incident to one previously colored black $C_{0}$-edge; note that these four
previously colored $C_{0}$-edges could not be colored by a previous
application of the inequality $1^{\ast }$. That is, before any application
of the recursion of the first inequality, one $C_{0}$-edge incident to each
of these four internal vertices of $P^{+}$ and $P^{-}$ has to be previously
colored black through an application of another recursion that is different
from the inequality $1^{\ast }$.

Now, in order to upper-bound the total number of instances produced at the
end of the algorithm, we will prove that, during the algorithm, in sufficiently
many produced instances we have sufficiently many edges which are \emph{not}
being forced by any of the recursions of the first two inequalities $1^{\ast}$ and $2^{\ast}$. To this
end, consider one application of the first inequality $f(x,y,k,\ell )\geq
2f(x-2,y-7,k-1,\ell )$, and assume that each of the four $C_{0}$-edges
incident to the four internal vertices of $P^{+}$ and $P^{-}$ is previously
colored black \emph{only} by applications of the second inequality $f(x,y,k,\ell
)\geq 2f(x-4,y-4,k,\ell -1)$. Furthermore, assume that the cycle $C_{6}$
with 6 edges that corresponds to the application of the first inequality
contains the four $C_{0}$-edges $v_{i-1}v_{i},v_{i}v_{i+1}$ and $%
v_{k-1}v_{k},v_{k}v_{k+1}$ (the second case where this cycle $C_{6}$ with 6
edges contains the four $C_{0}$-edges $v_{i}v_{i+1}$ and $%
v_{k-1}v_{k},v_{k}v_{k+1},v_{k+1}v_{k+2}$ can be analysed similarly). Then, 
this cycle $C_{6}$ either contains the two chords $v_{i-1}v_{k+1}$ and $%
v_{i+1}v_{k-1}$ or the two chords $v_{i-1}v_{k-1}$ and $v_{i+1}v_{k+1}$.
Moreover, by our assumption, each of the $C_{0}$-edges $v_{k+1}v_{k+2}$, $v_{k-1}v_{k-2}$, $%
v_{i-1}v_{i-2}$, and $v_{i+1}v_{i+2}$ has been previously colored black as a
part of the imperative path $P_{1}$, $P_{2}$, $P_{3}$, and $P_{4}$,
respectively, of an application of the second inequality. That is, $%
P_{1}=(v_{k+1},v_{k+2},v_{k+3},v_{k+4})$, $%
P_{2}=(v_{k-1},v_{k-2},v_{k-3},v_{k-4})$, $%
P_{3}=(v_{i-1},v_{i-2},v_{i-3},v_{i-4})$, and $%
P_{4}=(v_{i+1},v_{i+2},v_{i+3},v_{i+4})$. Note that, when the number $n$ of
vertices is large enough (in particular, when $n\geq 18$), at most one of
the paths $P_{1},P_{2}$ coincides with at most one of the paths $P_{3},P_{4}$%
. Furthermore, assume without loss of generality that the imperative path $%
P_{4}$ is forced \emph{after} all other imperative paths $P_{1},P_{2},P_{3}$
have been forced. Now we distinguish the two cases where (1) all imperative
paths $P_{1},P_{2},P_{3},P_{4}$ are disjoint and (2) $P_{2}=P_{4}$ (this
case is symmetric to the case where $P_{1}=P_{4}$).

\medskip

\emph{Case 1: all imperative paths }$P_{1},P_{2},P_{3},P_{4}$\emph{\ are
disjoint.} We assume without loss of generality that the cycle $C_{6}$ with
6 edges corresponding to the application of the first inequality contains
the two chords $v_{i-1}v_{k+1}$ and $v_{i+1}v_{k-1}$ (the case where this
cycle instead contains the two chords $v_{i-1}v_{k-1}$ and $v_{i+1}v_{k+1}$
is exactly symmetric). Then, at the time when the algorithm branches to
produce the second instance in the application of the second inequality at $%
P_{4}$, the imperative path $P_{4}$ is replaced either by the imperative
path $P_{4}^{\prime }=(v_{i},v_{i+1},v_{i+2},v_{i+3})$ or by the imperative
path $P_{4}^{\prime \prime }=(v_{i+2},v_{i+3},v_{i+4},v_{i+5})$; note that
both $P_{4}^{\prime }$ and $P_{4}^{\prime \prime }$ are still disjoint from $%
P_{1},P_{2},P_{3}$.

Suppose that $P_{4}$ is replaced by $P_{4}^{\prime }$. Then, the vertex $%
v_{y}$ at which the second inequality is applied (for forcing the two
alternative imperative paths $P_{4}$ and $P_{4}^{\prime }$) is $%
v_{y}=v_{i+2} $. Therefore, when the imperative path $P_{4}^{\prime }$ is
forced at $v_{i+2}$, the chord $v_{i+1}v_{k-1}$ is also forced to take tke
color blue. However, since the $C_{0}$-edge $v_{k-1}v_{k-2}$ is already
colored black (as part of the other imperative path $P_{2}$), the $C_{0}$%
-edge $v_{k}v_{k-1}$ is immediately forced to take the color red. That is,
the $C_{0}$-edge $v_{k}v_{k-1}$ is not being forced by either the first or
the second inequality.

Now suppose that $P_{4}$ is replaced by $P_{4}^{\prime \prime }$. Then the
five $C_{0}$-edges $%
v_{k}v_{k+1},v_{k}v_{k-1},v_{i-1}v_{i},v_{i}v_{i+1},v_{i+1}v_{i+2}$ are
still unolored, while their incident $C_{0}$-edges $%
v_{k+1}v_{k+2},v_{k-1}v_{k-2},v_{i-1}v_{i-2},v_{i+2}v_{i+3}$ are colored
black as part of the imperative paths $P_{1},P_{2},P_{3}$, and $%
P_{4}^{\prime \prime }$, respectively. Therefore it is not possible to force
(i.e.~color) any of these five uncolored $C_{0}$-edges by an application of
the first or the second inequality.

\medskip

\emph{Case 2: }$P_{2}=P_{4}$\emph{.} Note that in this case the endpoint $%
v_{k-1}$ of $P_{4}$ coincides with the endpoint $v_{i+4}$ of $P_{2}$, that
is, $v_{k-1}=v_{i+4}$. then, at the time when the algorithm branches to
produce the second instance in the application of the second inequality at $%
P_{4}$, the imperative path $P_{4}=P_{2}$ is replaced either by the
imperative path $P_{4}^{\prime
}=(v_{i},v_{i+1},v_{i+2},v_{i+3})=(v_{k-5},v_{k-4},v_{k-3},v_{k-2})$ or by
the imperative path $P_{4}^{\prime \prime
}=(v_{i+2},v_{i+3},v_{i+4},v_{i+5})=(v_{k-3},v_{k-2},v_{k-1},v_{k})$.

First assume that the cycle $C_{6}$ with 6 edges corresponding to the
application of the first inequality contains the two chords $v_{i-1}v_{k+1}$
and $v_{i+1}v_{k-1}$. Suppose that $P_{4}$ is replaced by $P_{4}^{\prime }$.
Then the four $C_{0}$-edges $%
v_{k}v_{k+1},v_{k}v_{k-1},v_{k-1}v_{k-2},v_{i-1}v_{i}$ are still unolored,
while their incident $C_{0}$-edges $%
v_{k+1}v_{k+2},v_{k-2}v_{k-3},v_{i-2}v_{i-1},v_{i}v_{i+1}$ are colored black
as part of the imperative paths $P_{1},P_{2}$, and $P_{4}^{\prime }$,
respectively. Therefore, as the chord $v_{i-1}v_{k+1}$ exists in the graph,
the only possibility where these four uncolored $C_{0}$-edges are forced by
an application of the first inequality is when the chord $%
v_{i}v_{k-2}=v_{i}v_{i+3}$ also exists in the graph. However, in this case
the two chords $v_{i}v_{k-2}$ and $v_{i+1}v_{k-1}$ forms an $X$-certificate
in the graph, which is a contradiction to our initial assumption in the
proof. Suppose that $P_{4}$ is replaced by $P_{4}^{\prime \prime }$ (instead
of $P_{4}^{\prime }$). Then it follows similarly that the chord $%
v_{i+2}v_{k} $ exists in the graph, which forms an $X$-certificate together
with the chord $v_{i+1}v_{k-1}$, which is again a contradiction.

Now assume that the cycle $C_{6}$ with 6 edges corresponding to the
application of the first inequality contains the two chords $v_{i-1}v_{k-1}$
and $v_{i+1}v_{k+1}$. Suppose that $P_{4}$ is replaced by $P_{4}^{\prime }$.
Then, the vertex $v_{y}$ at which the second inequality is applied (for
forcing the two alternative imperative paths $P_{4}$ and $P_{4}^{\prime }$)
is $v_{y}=v_{i+2}$. Therefore, when the imperative path $P_{4}^{\prime }$ is
forced at $v_{i+2}$, the chord $v_{i+1}v_{k+1}$ is also forced to take tke
color blue. However, since the $C_{0}$-edge $v_{k+1}v_{k+2}$ is already
colored black (as part of the other imperative path $P_{1}$), the $C_{0}$%
-edge $v_{k}v_{k+1}$ is immediately forced to take the color red. That is,
the $C_{0}$-edge $v_{k}v_{k+1}$ is not being forced by either the first or
the second inequality.

Finally suppose that $P_{4}$ is replaced by $P_{4}^{\prime \prime }$. Then,
the vertex $v_{y}$ at which the second inequality is applied (for forcing
the two alternative imperative paths $P_{4}$ and $P_{4}^{\prime }$) is $%
v_{y}=v_{i+3}$. Similarly to the above, when the imperative path $%
P_{4}^{\prime }$ is forced at $v_{i+3}$, the chord $v_{i-1}v_{k-1}$ is also
forced to take tke color blue. However, since the $C_{0}$-edge $%
v_{i-1}v_{i-2}$ is already colored black (as part of the other imperative
path $P_{3}$), the $C_{0}$-edge $v_{i}v_{i-1}$ is immediately forced to take
the color red. That is, the $C_{0}$-edge $v_{i}v_{i-1}$ is not being forced
by either the first or the second inequality.

\medskip

Our analysis in the above Cases 1 and 2 can be summarized as follows.
Consider an application of the first inequality $f(x,y,k,\ell )\geq
2f(x-2,y-7,k-1,\ell )$ at an instance $I_{0}$ that is produced during the
execution of the algorithm. This application of the first inequality
corresponds to a cycle $C_{6}$ with 6 edges which is the union of two paths $%
P^{+}$ and $P^{-}$; let $e_{1},e_{2},e_{3},e_{4}$ be the four (previously
colored black) $C_{0}$-edges incident to the four internal vertices of $%
P^{+} $ and $P^{-}$. Then, either:

\begin{itemize}
\item[(i)] at least one of $e_{1},e_{2},e_{3},e_{4}$ has not been previously
colored by any of the first two inequalities, or

\item[(ii)] each of $e_{1},e_{2},e_{3},e_{4}$ has been colored by previous
application of the second inequality but, in this case, \emph{before} the
last edge among $e_{1},e_{2},e_{3},e_{4}$ was colored, a new instance $%
I_{0}^{\prime }\neq I_{0}$ was also created, in which \emph{at least one}
additional $C_{0}$-edge is subsequently not forced by any application of the
first two inequalities.
\end{itemize}

Now recall that every application of the first inequality forces four new $C_{0}$-edges. 
Therefore, at every instance where the former case (i) (resp.~the latter case (ii)) applies, 
the number~$k$ of subsequent potential applications of the first inequality in both instances $I_{0}$ and $I_{0}^{\prime}$ 
(resp.~only in the instance $I_{0}^{\prime}$) decreases by at least $\frac{1}{4}$. Moreover, 
note that the latter case (ii) can apply the latest after four applications
of the second inequality. Therefore, we can replace \emph{every fourth}
application of the second inequality $2^{\ast }$ by the following inequality
which, in one of the produced instances, reduces the number $k$ of
subsequent potential applications of the inequality $1^{\ast }$ by $\frac{1}{4}$.

\begin{enumerate}
\item[$2^{\ast \ast }$.] $f(x,y,k,\ell )\geq f(x-4,y-4,k-\frac{1}{4},\ell
-1)+f(x-4,y-4,k,\ell -1)$.
\end{enumerate}

Thus we can also replace all applications of the inequalities $2^{\ast }$
and $2^{\ast \ast }$ by the next inequality which summarizes three
applications of the inequality $2^{\ast }$, followed by one application of
the inequality $2^{\ast \ast }$.

\begin{enumerate}
\item[$2^{\ast \ast \ast }$.] $f(x,y,k,\ell )\geq 8\left( f(x-16,y-16,k-%
\frac{1}{4},\ell -4)+f(x-16,y-16,k,\ell -4)\right) $.
\end{enumerate}

Finally, the variables $k,\ell $ have to satisfy the inequality $3k+\ell \leq \frac{n}{2}$, 
since every application of the inequality~$1^{\ast }$ (resp.~$%
2^{\ast }$) forces three new chords (resp.~one new chord), while the there
are in total $\frac{n}{2}$ chords in the graph.

Summarizing, we have to resolve the set $\{1^{\ast },2^{\ast \ast \ast
},3^{\ast },4^{\ast },5^{\ast },6^{\ast },7^{\ast },8^{\ast },9^{\ast }\}$
of inequalities, subject to the side constraint $3k+\ell \leq \frac{n}{2}$. To do so, we
set $f(x,y)=2^{\alpha x+\beta y+\gamma k+\delta \ell }$. Since $3k+\ell \leq 
\frac{n}{2}$, it follows that $\gamma
k+\delta \ell \leq \left( \gamma -3\delta \right) k+\frac{\delta }{2}n$.
Recall here that $k$ denotes the number of potential applications of the
inequality~$1^{\ast }$. Furthermore recall that any application of this
inequality forces (i.e.~colors) three chords of the graph, and thus $3k\leq 
\frac{n}{2}$. Therefore,
whenever $\gamma -3\delta \geq 0$, it follows that $\gamma k+\delta \ell
\leq \frac{\gamma }{6}n$. Moreover, since $x\leq n$ and $y\leq \frac{3}{2}n$%
, it follows that $\alpha x+\beta y+\gamma k+\delta \ell \leq (\alpha +\frac{%
3}{2}\beta +\frac{\gamma }{6})n$. As it can be easily verified, the values $%
\{\alpha =0.151600116,\ \beta =0.096388,\gamma =0.022083768,\delta
=0.007361256\}$ satisfy all these inequalities together. Therefore, as $%
\alpha +\frac{3}{2}\beta +\frac{\gamma }{6}=0.299862744$, it follows that
the total number of instances produced by the algorithm is upper-bounded 
by $2^{0.299862744\cdot n} = O(1.23103^n)$. This completes the analysis of Part~C.

\medskip

To conclude with upper-bounding the total running time of Algorithm~\ref{alt-cycle-detection-alg}, 
assume that we have already computed all the desired alternating
red-blue \emph{cycles} $D=Red\cup Blue$ of the input instance $I$, each of
which, using the ambivalent quadruples, can potentially encode many other
alternating red-blue cycles. Initially, in lines~\ref{alg-line-X-certif-1}-\ref{alg-line-X-certif-2} the
algorithm searches for all $X$-certificates, which can be trivially done in $%
O(n^{2})$ time by just examining every pair among the $\frac{n}{2}$ chords of $G$. 
As we have upper-bounded the total number of instances that
the algorithm produces, it remains to compute the running time of each
execution of Procedure~\ref{ambivalent-proc}. The symmetric difference $%
C_{0}\ \Delta \ D$ can be computed in linear $O(n)$ time. If $C_{0}\ \Delta
\ D$ is connected, then this is the second Hamiltonian cycle that the
algorithm outputs. Assume that $C_{0}\ \Delta \ D$ has $k\geq 2$ connected
components. Then we create in linear $O(n)$ time a new auxiliary graph $%
H=(V_{H},E_{H})$, as follows. The vertex set $V_{H}$ has one vertex for
every connected component (i.e.~cycle) of $C_{0}\ \Delta \ D$, and thus $%
|V_{H}|\leq n$. Let $u_{1},u_{2}\in V_{H}$ be two different vertices of $H$,
i.e.~corresponding to two different connected components of $C_{0}\ \Delta \
D$. Then, $u_{1}$ is adjacent to $u_{2}$ in $E_{H}$ if and only if there
exists an ambivalent quadruple which, if flipped, will connect the two
corresponding connected components of $C_{0}\ \Delta \ D$. Lines~\ref%
{proc-ambi-line-1}-\ref{proc-ambi-line-3} of Procedure~\ref{ambivalent-proc}
can be implemented as follows. We run any linear-time (i.e.~$O(n)$-time)
connectivity algorithm on $H$ such as Breadth-First-Search. If $H$ is not
connected then no sequence of flips of the ambivalent quadruples can connect
the components of $C_{0}\ \Delta \ D$ into one Hamiltonian graph of $G$.
Otherwise we obtain a spanning tree of $H$, and in this case the edges of
the spanning tree indicate those ambivalent quadruples that need to be
flipped in order to make $C_{0}\ \Delta \ D$ a Hamiltonian graph of $G$.

Finally, it is easy to see that the space complexity of the algorithm is linear, i.e.~$O(n)$. 
Indeed, for every instance that is produced by the algorithm using the recursions, 
we only need to keep in memory the colors of the $\frac{3}{2}n$ edges and a linear number of 
ambivalent edge quadruples.
\end{proof}

\section{Efficiently computing another long cycle in a Hamiltonian graph}
\label{long-cycle-sec}

In this section we prove our results on approximating the length of a second cycle on graphs with minimum degree $\delta\geq 3$ and maximum degree $\Delta$. 
In~\cite{bazgan1999approximation}, Bazgan, Santha, and Tuza considered the optimization problem of 
efficiently (i.e.~in polynomial time) constructing a large second cycle different than the given Hamiltonian cycle $C_0$ in a given Hamiltonian graph $G$. 
In particular they proved the following results. 

\begin{theorem}[\hspace{-0,001cm}\cite{bazgan1999approximation}]\label{th:bazgan-eptas}
	Let $G$ be an $n$-vertex cubic Hamiltonian graph and let $C_0$ be a Hamiltonian cycle of $G$. 
	Given $G$ and $C_0$, for every $\varepsilon > 0$, a cycle $C' \neq C_0$ of length at least 
	$(1-\varepsilon)n$ can be found in time $2^{O(1/\varepsilon^2)} \times n$.
\end{theorem}

\begin{theorem}[\hspace{-0,001cm}\cite{bazgan1999approximation}]\label{th:bazgan-approx}
	Let $G$ be an $n$-vertex cubic Hamiltonian graph and let $C_0$ be a Hamiltonian cycle of $G$. 
	There is an algorithm which, given $G$ and $C_0$, computes a cycle $C' \neq C_0$ of length
	at least $n - 4\sqrt{n}$ in time $O(n^{3/2} \log{n})$.
\end{theorem}

\subsection{Notation and preliminary results}
Before we proceed to the main result of the section, we introduce some necessary notation and state preliminary results.
Let $G=(V,E)$ be a graph with a designated Hamiltonian cycle $C_0 = (v_1,v_2, \ldots, v_n,v_1)$.
Two chords of $C_0$ are \emph{independent} if they do not share an endpoint.
The \emph{length} of a chord $v_iv_j$, with $i < j$, is defined as $\min \{ j-i, n+i-j \}$.
We say that two vertices $u,v \in V$ are \emph{chord-adjacent} if they are
connected by a chord of $G$.
Two independent chords $e_1$ and $e_2$ are called \emph{crossing} if their
endpoints appear in an alternating order around~$C_0$; otherwise $e_1$ and $e_2$ are called \emph{parallel}.

For $x,y \in V$, we denote by $d(x,y)$ the length of the path from $x$ to $y$
around $C_0$. Note that, in general, $d(x,y) \neq d(y,x)$.
We define the \emph{distance} between two independent chords $xy$ and $ab$
as follows:
\begin{enumerate}
	\item if $xy$ and $ab$ are crossing, such that $a$ lies on the path from $x$ to $y$ around $C_0$, then 
	$\dist(xy,ab) = \min\{ d(x,a) + d(y,b), d(b,x) + d(a,y) \}$;
	
	\item if $xy$ and $ab$ are parallel such that neither $y$ nor $b$ 
	lie on the path from $x$ to $a$ around $C_0$, then
	$\dist(xy,ab) = d(x,a) + d(b,y)$.
\end{enumerate}

In the proof of our main result of this section (see Theorem~\ref{th:longCycle}) we use the following two lemmas. 
The first one is a basic fact from graph theory and the second one is straightforward to check (see \cref{fig:Lem4} for an illustration).

\begin{lemma}\label{lem:Matching}[\hspace{-0,001cm}\cite{west}, Exercise 3.1.29]
	Let $G=(V,E)$ be a bipartite graph of maximum degree~$\Delta$.
	Then $G$ has a matching of size at least $\frac{|E|}{\Delta}$.
\end{lemma}

\begin{lemma}\label{lem:chords}
	Let $G = (V,E)$ be an $n$-vertex graph with a Hamiltonian cycle $C_0$.
	\begin{enumerate}
		\item[(1)] If $G$ has a chord of length $\ell$, then $G$ contains a cycle $C' \neq C_0$ of length
		at least $n-\ell+1$.

		\item[(2)] 
		If $G$ has two crossing chords $e_1$, $e_2$ and $\dist(e_1,e_2) = d$, then
		$G$ contains a cycle $C' \neq C_0$ of length at least $n-d+2$.
		
		\item[(3)] 
		If $G$ has four pairwise independent chords $e_1$, $e_2$, $f_1$, and $f_2$ such that 
		\begin{enumerate}
			\item $e_1$, $e_2$ are parallel and $f_1$, $f_2$ are parallel, 
			\item $e_i$ and $f_j$ are crossing for every $i,j \in \{1,2\}$, 
			\item $\dist(e_1,e_2) = d_1$ and  $\dist(f_1,f_2) = d_2$,
		\end{enumerate}
		then $G$ contains a cycle $C' \neq C_0$ of length at least $n - d_1 - d_2 + 4$.
	\end{enumerate}
\end{lemma}

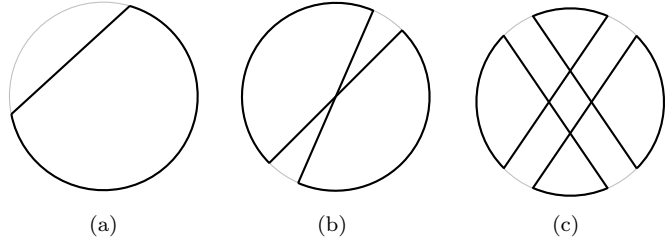
\begin{figure}[h!]
\centering%
\subfigure[]{
\begin{tikzpicture}[scale=0.5]
	\draw[gray!50] (0,0) circle (2.5);
	
	\draw[black, thick] (-2.45, -0.48) arc (191:434:2.5);
	
   	\draw[thick] (-2.45, -0.48) -- (0.7, 2.4);
\end{tikzpicture}
}
\subfigure[]{
\begin{tikzpicture}[scale=0.5]
	\draw[gray!50] (0,0) circle (2.5);
	
	\draw[black, thick] (1.761, 1.761) arc (45:-113.57:2.5);
	\draw[black, thick] (1.003, 2.285) arc (66.5:224.95:2.5);
	
   	\draw[thick] (-1.772, -1.772) -- (1.772, 1.772);
   	\draw[thick] (-1, -2.3) -- (1, 2.3);
\end{tikzpicture}
}
\subfigure[]{
\begin{tikzpicture}[scale=0.5]
	\draw[gray!50] (0,0) circle (2.5);
	
	\draw[black, thick] (1.761, 1.761) arc (45:-44.57:2.5);
	\draw[black, thick] (1.003, 2.285) arc (66.5:113.7:2.5);
	
	\draw[black, thick] (-1.761, -1.761) arc (-135:-225:2.5);
	\draw[black, thick] (1.003, -2.285) arc (-66.5:-113.7:2.5);
	
   	\draw[thick] (-1.772, -1.772) -- (1, 2.3);
   	\draw[thick] (-1, -2.3) -- (1.772, 1.772);
   	
   	\draw[thick] (-1.772, 1.772) -- (1, -2.3);
   	\draw[thick] (-1, 2.3) -- (1.772, -1.772);
\end{tikzpicture}
}
\caption{An illustration of \cref{lem:chords}. (a)~A short chord. (b)~A pair of crossing chords. (c)~Crossing pairs of parallel chords.}
\label{fig:Lem4}
\end{figure}

\subsection{Long cycles in Hamiltonian graphs}

\begin{theorem}\label{th:longCycle}
	Let $G = (V,E)$ be an $n$-vertex Hamiltonian graph of minimum degree $\delta \geq 3$.
	Let $C_0=(v_1,v_2, \ldots, v_n,v_1)$ be a Hamiltonian cycle in $G$ and 
	let $\Delta$ denote the maximum degree of $G$. 
	Then $G$ has a cycle $C' \neq C_0$ of length at least $n-4\alpha(\sqrt{n} + 2\alpha)+8$, where
	$\alpha = \frac{\Delta-2}{\delta-2}$. 
	Moreover, given $C_0$, such a cycle $C'$ can be computed in $O(m)$ time, where $m = |E|$.
\end{theorem}
\begin{proof}
	We start by showing the existence of the desired cycle $C'$. 
	Without loss of generality we assume that 
	$\alpha < \frac{\sqrt{n}}{2}$, as otherwise any cycle $C' \neq C_0$ in $G$ satisfies the theorem.
	Furthermore, we assume that the length of every chord in $G$ is at least 
	$4 \alpha (\sqrt{n} + 2\alpha) - 6$,
	as otherwise the existence of $C'$ follows from \cref{lem:chords} (1). 
	
	Let $q = \alpha \sqrt{n}$. 
	We arbitrarily partition the vertices\footnote{More formally, we partition the interval $[1,n]$ 
	into the consecutive intervals $B_0, B_1, \ldots, B_{r-1}$, 
	which immediately implies a partition of the vertices of the Hamiltonian cycle $C_0$.} 
	of the Hamiltonian cycle $C_0$ into $r$ consecutive intervals 
	$B_0, B_1, \ldots, B_{r-1}$, such that 
	$r \in \Big\{ \Big\lfloor \frac{\sqrt{n}}{\alpha} \Big\rfloor, \Big\lfloor \frac{\sqrt{n}}{\alpha} \Big\rfloor + 1 \Big\}$ and 
	$\lfloor q \rfloor \leq |B_i| \leq \lfloor q \rfloor + 2\alpha^2$ for every 
	$i \in \{ 0, 1, \ldots, r-1 \}$. It is a routine task to check that such a partition exists.

	For every $i \in \{ 0, 1, \ldots, r-1 \}$ we denote by $W_i$ the set of vertices that are chord-adjacent
	to a vertex in $B_i$, and by $E_i$ we denote the set of chords that are incident to a vertex in $B_i$.
	Furthermore, we denote by $H_i$ the graph with vertex set $B_i \cup W_i$ and edge set $E_i$.
	Since the length of every chord in $G$ is at least $4 \alpha (\sqrt{n} + 2\alpha) - 6$, observe that
	for every $i \in \{ 0, 1, \ldots, r-1 \}$, the set $W_i$ is disjoint from $B_{i-1} \cup B_i \cup B_{i+1}$
	(where the arithmetic operations with indices are modulo $r$).
	The latter, in particular, implies that $H_i$ is a bipartite graph with color classes 
	$B_i$ and $W_i$.
	
 	Let $i,j \in \{ 0, 1, \ldots, r-1 \}$ be two distinct indices, we say that the intervals $B_i$ and $B_j$
 	are \emph{matched} if there exist two independent chords such that each of them has one endpoint 
 	in $B_i$ and the other endpoint in $B_j$. 
 	We claim that every interval $B_i$ is matched to another interval $B_j$
 	for some $j \in \{ 0, 1, \ldots, r-1 \} \setminus \{ i-1, i, i+1 \}$.
 	Indeed, by \cref{lem:Matching}, graph $H_i$ has a matching $M_i$ of size at least 
 	$$
 		\frac{\lfloor q \rfloor (\delta-2)}{\Delta-2} 
 		= \frac{\lfloor \alpha \sqrt{n} \rfloor}{\alpha} 
 		> \frac{\alpha \sqrt{n} - 1}{\alpha}
 		\geq \sqrt{n} - 1
 		> \Big\lfloor \frac{\sqrt{n}}{\alpha} \Big\rfloor - 2 
 		\geq r-3,
 	$$
 	and therefore, by the pigeonhole principle, there exists $j \in \{ 0, 1, \ldots, r-1 \} \setminus \{ i-1, i, i+1 \}$
 	such that at least two edges in $M_i$ have their endpoints in $B_j$, meaning that
 	$B_i$ is matched to $B_j$.
 	
 	Let $\sigma : \{ 0, 1, \ldots, r-1 \} \rightarrow \{ 0, 1, \ldots, r-1 \}$ be a function such that $B_i$ is
 	matched to $B_{\sigma(i)}$, and denote by $f_{i,1}$ and $f_{i,2}$ some fixed pair of independent
 	chords between $B_i$ and $B_{\sigma(i)}$. 
 	We observe that 
 	$\dist(f_{i,1}, f_{i,2}) \leq 2 \big( \lfloor q \rfloor + 2\alpha^2 - 1 \big) \leq 2\alpha(\sqrt{n} + 2\alpha) - 2$,
 	as the endpoints of $f_{i,1}$ and $f_{i,2}$ lie in the intervals $B_i$ and $B_{\sigma(i)}$ 
 	each of length at most $\lfloor q \rfloor + 2\alpha^2$.
 	
 	Let now $R$ be an auxiliary graph with a Hamiltonian cycle $(x_0,x_1, \ldots, x_{r-1})$ and the
 	chord set being $\{ x_i x_{\sigma(i)} : i = 0, 1, \ldots, r-1 \}$. 
	Let $x_ix_j$ be a chord in $R$ of the minimum length, where $j = \sigma(i)$. 
	Without loss of generality, we assume
	that $i < j$ and $j-i \leq r+i-j$. Let $x_k$ be a vertex of $R$ such that $i < k < j$ and let $s = \sigma(k)$.
	Since $x_ix_j$ is of minimum length, the chords $x_i x_j$ and $x_k x_s$ are crossing, and hence
	each of $f_{i,1}$ and $f_{i,2}$ crosses both $f_{k,1}$ and $f_{k,2}$.
	
	Finally, if $f_{i,1}$, $f_{i,2}$ or $f_{k,1}$, $f_{k,2}$ are crossing, then
	by \cref{lem:chords} (2) there exists a cycle $C' \neq C_0$ of length at least 
	$n-2\alpha(\sqrt{n} + 2\alpha)+4$.
	Otherwise, $f_{i,1}, f_{i,2}$ are parallel and $f_{k,1}, f_{k,2}$ are parallel, and hence by 
	\cref{lem:chords} (3) there exists a cycle $C' \neq C_0$ of length at least 
	$n-4\alpha(\sqrt{n} + 2\alpha)+8$, which proves the first part of the theorem.

	The above proof is constructive. We now explain at a high level how the proof can be turned 
	into the desired algorithm.
	First, if $\alpha \geq \frac{\sqrt{n}}{2}$,
	then we output any cycle formed by a chord and the longer path of $C_0$ connecting the endpoints
	of the chord.
	Otherwise, we partition the vertices of $C_0$ into the intervals $B_1, \ldots, B_{r-1}$
	and we assign to each vertex the index of its interval. 
	Clearly, this can be done in $O(n)$ time.
	Next, we traverse the vertices of $G$ along the cycle $C_0$ and for every vertex $v$ of an interval $B_i$
	we check the chords incident to $v$. If we encounter a chord $f$ of length less than
	$4 \alpha (\sqrt{n} + 2\alpha) - 6$, then we output the cycle formed by $f$ and the longer path of $C_0$
	connecting the endpoints of $f$. Otherwise, for the interval $B_i$ we keep the information of 
	how many and which vertices of $W_i$ belong to other intervals $B_j$ for 
	$j \in \{ 0, 1, \ldots, r-1 \} \setminus \{ i-1, i, i+1 \}$. When we find an interval $B_j$ that has
	at least two elements from $W_i$, we set $\sigma(i)$ to $j$ and proceed to the first vertex 
	of the next interval $B_{i+1}$. By doing this, we also keep the information of the current 
	shortest chord in the graph $R$ (defined in the proof above). 
	After finishing this procedure: (1)~we have a function $\sigma(\cdot)$;
	(2)~for every $i \in \{ 0, 1, \ldots, r-1 \}$ we know a pair $f_{i,1}$, $f_{i,2}$ of independent edges between 
	$B_i$ and $B_{\sigma(i)}$; and 
	(3)~we know $k$ such that $x_kx_{\sigma(k)}$ is a minimum length chord in $R$. 
	Clearly, this information is enough to identify the desired cycle in constant time.
	In total, we spent $O(n)$ time to compute the partition of the vertices into the intervals and
	we visited every chord at most twice, which implies the claimed $O(m)$ running time.
\end{proof}

\medskip

The next two corollaries are implied as immediate consequences of Theorem~\ref{th:longCycle}, and they provide immediate extensions of the results 
of~\cite{bazgan1999approximation} and~\cite{girao2019long}, respectively. 

\begin{corollary}\label{cor:regular}
	Let $G = (V,E)$ be an $n$-vertex Hamiltonian $\delta$-regular graph with $\delta \geq 3$, and
	let $C_0$ be a Hamiltonian cycle of $G$. 
	Then $G$ has a cycle $C' \neq C_0$ of length at least $n - 4\sqrt{n}$, 
	which can be computed in $O(\delta n)$ time.
\end{corollary}

\begin{corollary}\label{cor:squareRation}
	Let $G = (V,E)$ be an $n$-vertex Hamiltonian graph of minimum degree $\delta \geq 3$.
	Let $C_0$ be a Hamiltonian cycle of $G$ and let $\Delta$ denote the maximum degree of $G$. 
	If $\frac{\Delta}{\delta} = o(\sqrt{n})$, then $G$ has a cycle $C' \neq C_0$ of length at least $n-o(n)$,
	which can be computed in $O(m)$ time.
\end{corollary}


\begin{thebibliography}{10}

\bibitem{bazgan1999approximation}
C.~Bazgan, M.~Santha, and Z.~Tuza.
\newblock On the approximation of finding a(nother) hamiltonian cycle in cubic
  hamiltonian graphs.
\newblock {\em Journal of Algorithms}, 31(1):249--268, 1999.

\bibitem{Bel62}
R.~Bellman.
\newblock Dynamic programming treatment of the {Travelling Salesman Problem}.
\newblock {\em Journal of the ACM}, 9(1):61--63, 1962.

\bibitem{B14}
A.~Bjorklund.
\newblock Determinant sums for undirected hamiltonicity.
\newblock {\em SIAM Journal on Computing}, 43(1):280--299, 2014.

\bibitem{BjorklundH13}
A.~Bj{\"{o}}rklund and T.~Husfeldt.
\newblock The parity of directed {Hamiltonian} cycles.
\newblock In {\em Profeecings of the 54th Annual {IEEE} Symposium on
  Foundations of Computer Science ({FOCS})}, pages 727--735, 2013.

\bibitem{BjorklundHKK12}
A.~Bj{\"{o}}rklund, T.~Husfeldt, P.~Kaski, and M.~Koivisto.
\newblock The traveling salesman problem in bounded degree graphs.
\newblock {\em {ACM} Transactions on Algorithms}, 8(2):18:1--18:13, 2012.

\bibitem{BHKN}
H.~L. Bodlaender, M.~Cygan, S.~Kratsch, and J.~Nederlof.
\newblock Deterministic single exponential time algorithms for connectivity
  problems parameterized by treewidth.
\newblock {\em Information and Computation}, 243:86--111, 2015.

\bibitem{Cameron01}
K.~Cameron.
\newblock Thomason's algorithm for finding a second {Hamiltonian} circuit
  through a given edge in a cubic graph is exponential on {Krawczyk's} graphs.
\newblock {\em Discrete Mathematics}, 235:69--77, 2001.

\bibitem{CyganKN13}
M.~Cygan, S.~Kratsch, and J.~Nederlof.
\newblock Fast hamiltonicity checking via bases of perfect matchings.
\newblock In {\em Proceedings of the 45th {ACM} Symposium on Theory of
  Computing Conference ({STOC})}, pages 301--310, 2013.

\bibitem{entringer-swart80}
R.~Entringer and H.~Swart.
\newblock Spanning cycles of nearly cubic graphs.
\newblock {\em Journal of Combinatorial Theory, Series B}, 29(3):303--309,
  1980.

\bibitem{Eppstein07}
D.~Eppstein.
\newblock The {Traveling Salesman Problem} for cubic graphs.
\newblock {\em Journal of Graph Algorithms and Applications}, 11(1):61--81,
  2007.

\bibitem{fleischner94}
H.~Fleischner.
\newblock Uniqueness of maximal dominating cycles in 3-regular graphs and of
  {Hamiltonian} cycles in 4-regular graphs.
\newblock {\em Journal of Graph Theory}, 18(5):449--459, 1994.

\bibitem{girao2019long}
A.~Gir{\~a}o, T.~Kittipassorn, and B.~Narayanan.
\newblock Long cycles in {Hamiltonian} graphs.
\newblock {\em Israel Journal of Mathematics}, 229(1):269--285, 2019.

\bibitem{G14}
R.~J. Gould.
\newblock Recent advances on the {Hamiltonian} problem: {Survey III}.
\newblock {\em Graphs and Combinatorics}, 30(1):1--46, 2014.

\bibitem{HK62}
M.~Held and R.~M. Karp.
\newblock A dynamic programming approach to sequencing problems.
\newblock {\em Journal of the Society for Industrial and Applied Mathematics},
  10(1):196--210, 1962.

\bibitem{K72}
R.~M. Karp.
\newblock Reducibility among combinatorial problems.
\newblock In {\em Complexity of Computer Computations}, pages 85--103.
  Springer, 1972.

\bibitem{Krawczyk99}
A.~Krawczyk.
\newblock The complexity of finding a second {Hamiltonian} cycle in cubic
  graphs.
\newblock {\em Journal of Computer and System Sciences}, 58(3):641--647, 2001.

\bibitem{LS14}
M.~Li{\'s}kiewicz and M.~R. Schuster.
\newblock A new upper bound for the {Traveling Salesman Problem} in cubic
  graphs.
\newblock {\em Journal of Discrete Algorithms}, 27:1--20, 2014.

\bibitem{pap94}
C.~H. Papadimitriou.
\newblock On the complexity of the parity argument and other inefficient proofs
  of existence.
\newblock {\em Journal of Computer and system Sciences}, 48(3):498--532, 1994.

\bibitem{T78}
A.~G. Thomason.
\newblock Hamiltonian cycles and uniquely edge colourable graphs.
\newblock {\em Advances in Graph Theory}, 3:259--268, 1978.

\bibitem{Tu46}
W.~T. Tutte.
\newblock On {Hamiltonian} circuits.
\newblock {\em Journal of the London Mathematical Society}, 1(2):98--101, 1946.

\bibitem{west}
D.~West.
\newblock {\em Introduction to graph theory}.
\newblock Prentice hall Upper Saddle River, 2 edition, 2001.

\bibitem{XiaoNagamochi16}
M.~Xiao and H.~Nagamochi.
\newblock An exact algorithm for {TSP} in degree-3 graphs via circuit procedure
  and amortization on connectivity structure.
\newblock {\em Algorithmica}, 74(2):713--741, 2016.

\end{thebibliography}
\end{document}